\documentclass[11pt]{article}
\usepackage[utf8]{inputenc}
\usepackage[english]{babel}
\usepackage{graphicx}
\usepackage{amsmath}
\usepackage{amssymb}
\usepackage{float}
\usepackage{caption}
\usepackage{amsthm}
\usepackage{xcolor}
\usepackage{tabu}
\usepackage{fullpage}
\usepackage{hyperref}
\hypersetup{colorlinks,linkcolor={red!50!black},citecolor={blue!50!black},urlcolor={blue!80!black}}
\usepackage{cleveref}
\usepackage{xspace}
\usepackage{thm-restate}

\iffalse
\renewbibmacro*{doi+eprint+url}{
    \printfield{doi}
    \newunit\newblock
    \iftoggle{bbx:eprint}{
        \usebibmacro{eprint}
    }{}
    \newunit\newblock
    \iffieldundef{doi}{
        \usebibmacro{url+urldate}}
        {}
    }
\fi

\theoremstyle{plain}
\newtheorem{theorem}{Theorem}
\newtheorem{lemma}[theorem]{Lemma}

\theoremstyle{definition}

\title{Two Tiling is Undecidable}
\author{Jack Stade}
\date{November 2026}

\begin{document}

\maketitle

\begin{abstract}
We show that the following problem is undecidable: given two polygonal prototiles, determine whether the plane can be tiled with rotated and translated copies of them. This improves a result of Demaine and Langerman [SoCG 2025], who showed undecidability for three tiles.

Along the way, we show that tiling with one prototile is undecidable if there can be edge-to-edge matching rules. This is the first result to show undecidability for monotiling with only local matching constraints.
\end{abstract}

\tableofcontents

\section{Introduction}

A tiling problem (in $\mathbb{R}^2$) asks whether a given set of prototiles can tile the plane, that is whether there is a tiling of the plane where every tile is a copy of one of the prototiles. Rotations and reflections may or may not be allowed, and there may be some constraints on how tiles can be placed in relation to each other.

A classic problem asks which tiling problems are decidable, that is, when is there an algorithm that can determine whether a set of prototiles admits a tiling? In 1966, Berger \cite{BergerWangTiling} showed that the \emph{Wang tiling problem} is undecidable. This is a tiling problem where the prototiles are squares with colored edges, where the tiles should be placed without rotations on a square grid so that colors match along coincident edges. This type of matching rule can be simulated with purely geometric tiles (see \cite{GolombSimulatingWangTiles}). Several different proofs of the undecidability of Wang tiling are now known (see e.g. \cite{KariWangTiling,OllingerWangTiling,RobinsonWangTiling}).

There are only finitely many meaningfully distinct sets of $n$ Wang tiles for any constant $n$, so any undecidability result for Wang tiling must allow the number of tiles to grow arbitrarily large. For more general tiles, this need not be the case. In 2009, Ollinger \cite{Ollinger5Tiling} showed that tiling (with rotations) with $5$ polyomino prototiles is undecidable. For general (not necessarily polyomino) tiling, Demaine and Langerman \cite{DemaineLangerman3Tiling} recently showed that $3$ prototiles are sufficient for undecidability. We improve this result, showing that geometric tiling is undecidable with $2$ prototiles:

\begin{theorem}\label{thm:main}
There is no algorithm, that, given two polygonal prototiles, decides whether they tile the plane. 
\end{theorem}

A longstanding open problem \cite{GoodmanStraussSurvey} asks whether \emph{monotiling} is decidable. We make some progress towards answering this question, showing that monotiling is undecidable when certain additional constraints are enforced. An \emph{edge-to-edge matching rule} is a constraint that prevents certain pairs of edges of prototiles from being placed next to each other in a tiling. We show that:

\begin{theorem}\label{thm:onetiling}
There is no algorithm, that, given a single prototile with edge-to-edge matching rules, decides whether it tiles the plane.
\end{theorem}

Edge-to-edge matching rules are in some sense ``local''. When non-local ``atlas-style'' matching rules are allowed, it is straightforward (see \cite{GoodmanStraussSurvey}) to see that monotiling is undecidable. Our result is the first to show undecidability for monotiling with only local matching rules.

\Cref{thm:onetiling} is the main technical contribution of the paper. We first prove \Cref{thm:onetiling} in \Cref{sec:monotiling}, and in \Cref{sec:stapletile} we show how to obtain \Cref{thm:main} by adding a ``staple'' tile that simulates edge-to-edge matching rules. This is similar to the use of a staple tile in \cite{DemaineLangerman3Tiling}.

Both \Cref{thm:main} and \Cref{thm:onetiling} are stated in the setting where rotations are allowed but reflections are prohibited. It is not too hard to modify the tiles (see \Cref{sec:stapletile}) so that they work even if reflections are allows. The modifications prevent an edge of a reflected tile from being placed against an edge of a non-reflected tile.

One could also ask what happens when only translations (no rotations or reflections) are allowed. Ollinger's original construction \cite{Ollinger5Tiling} used $11$ different orientations of the $5$ tiles, showing that translational $11$-tiling is undecideable. After a sequence of incremental improvements (\cite{Yang10Tiling,Yang9Tiling,YangZhang8Tiling}), Yang and Zhang \cite{YangZhang7Tiling} showed that translational $7$-tiling is undecidable. Concurrently to the present work, Kim \cite{KimPolyominos} showed that translational $5$-tiling is undecidable. Our construction uses only $4$ total orientations of our $2$ tiles, showing that:

\begin{theorem}\label{thm:translationaltiling}
Translational $4$-tiling is undecidable. 
\end{theorem}

Translational monotiling is known to be decidable (see \cite{GBNTranslationalMonotiling}), so only the cases $n=2$ and $n=3$ remain open.

Some results restrict to tiling with polyomino tiles. Our tiles are not polyominos (though our construction can be realized with polyhexes, see \Cref{sec:variants}). In \Cref{sec:variants}, we also sketch a proof that tiling with $3$ polyomino prototiles is undecidable. This improves Ollinger's construction with $5$ tiles \cite{Ollinger5Tiling} and matches Kim's construction with $3$ \cite{KimPolyominos}. Recently, Abel, Akitaya, Chung, Demaine, Diomidova, Hendrickson, Langerman and Lynch showed that polyomino $2$-tiling is undecideable, but only if one of the prototiles can be disconnected. 

\Cref{tab:variantstable} gives a summary of tiling problems where known upper bounds for undecidability can be improved by our techniques. 

\begin{table}
    \centering
    \begin{tabular}{|ccc|c|c|}
    \hline
    Prototiles &  Rotations & Matching rules & Previous Best & Our Construction \\
    \hline
    Polygons & Yes & No & $3$ \cite{DemaineLangerman3Tiling} & $2$ \\
    Polyhexes & Yes & No & n/a & $2$\\
    Polyominos & Yes & No & $5$ \cite{Ollinger5Tiling}, $3^*$ \cite{KimPolyominos} & $3$ \\
    Polyominos & No & No & $7$ \cite{YangZhang7Tiling}, $5^*$ \cite{KimPolyominos} & $4$ \\ 
    Polygons & Yes & Yes & $2$ \cite{DemaineLangerman3Tiling} & $1$ \\
    Polyhexes & Yes & Yes & n/a & $1$\\
    Polyominos & Yes & Yes & $2^*$ \cite{KimPolyominos} & $2$ \\
    Polyominos & No & Yes & $4^*$ \cite{KimPolyominos} & $3$ \\
    \hline
    \end{tabular}
    \caption{Tiling problems that can be shown to be undecidable with our techniques and other known bounds on undecidability. Entries marked with $*$ come from concurrent work by Kim \cite{KimPolyominos}}
    \label{tab:variantstable}
\end{table}

\paragraph{Lean formalization.}

A Lean formalization of the proof of Theorem 2 is provided in

\noindent
\href{https://github.com/JackStade/stick-tiling-lean}{https://github.com/JackStade/stick-tiling-lean}.

\section{Monotiling with edge-to-edge matching rules is undecidable}\label{sec:monotiling}

The goal of this section is to prove \Cref{thm:onetiling}. 

\subsection{The weave pattern}\label{sec:weavepattern}

The prototile that we use is a \emph{stick} formed by $n$ hexagons. An example is shown in \Cref{fig:sticktile}. Since reflections are not allowed, an arrow is sufficient to specify the orientation of the tile.
The edges of the tile are given labels as shown in the figure.
The edges $z_1a_1b_1\ldots b_{n-1}x_2$ which form the part behind the arrow are called the \emph{back side} and the edges $z_2c_1d_1\ldots c_{n-1}d_{n-1}x_1$ are called the \emph{front side}.

\begin{figure}
    \centering
    \includegraphics[page=1]{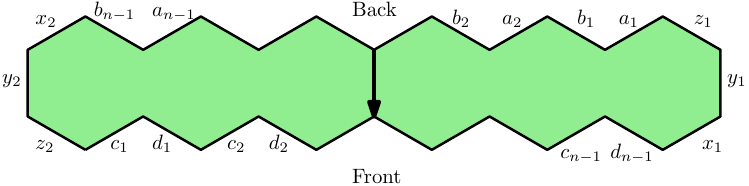}
    \caption{A stick of length $6$.}
    \label{fig:sticktile}
\end{figure}

There are many different ways to tile the plane with stick tiles. Our first step is to add some edge-to-edge matching rules that restrict the tile to forming a well-structured type of tiling that we call a \emph{weave pattern}. In order to define a weave pattern, we first need some definitions.

Given a tiling of the plane by stick tiles, we can always orient it so that one of the three classes of parallel line segments on the boundary of a stick is vertical. Up to rotations by $\pi$, there are $3$ orientations that a stick might occur in. We say that a stick is either horizontal, slants left, or slants right, as shown in \Cref{fig:stickorientations}.

\begin{figure}
    \centering
    \includegraphics[page=2]{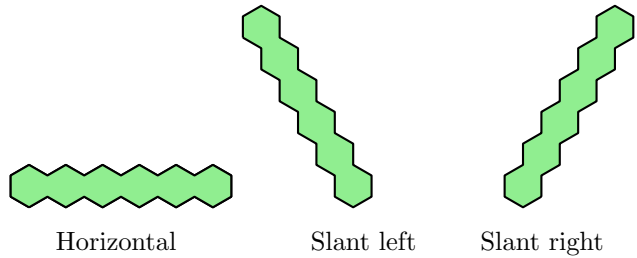}
    \caption{The orientations of a stick (up to rotations by $\pi$)}
    \label{fig:stickorientations}
\end{figure}

We say that a pair of sticks are \emph{stacked} if they have the same orientation modulo $\pi$ and one stick fills all of the concave corners on one side of the other stick. A pair of stacked sticks leans left if it (up to rotation) as in the left side of \Cref{fig:stackdirections} and leans right if it as in the right side of the figure. A stack of sticks is formed by a chain of stacked pairs.

\begin{figure}
    \centering
    \includegraphics[page=5]{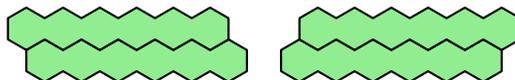}
    \caption{Stacked pairs leaning left and right. The direction that a pair leans is invariant under rotation.}
    \label{fig:stackdirections}
\end{figure}

The structure that we want to create is shown in \Cref{fig:weavepattern}. Specifically, we say that a tiling by stick polyhexes forms a \emph{weave pattern} if it can be oriented so that:

\begin{itemize}
    \item All the sticks are either horizontal or slant left
    \item There is a horizontal stick
    \item All the horizontal sticks are in left-leaning stacks (called the horizontal stacks) and all the left-slanted sticks are in right-leaning stacks (called the vertical stacks)
    \item All the horizontal stacks have height $n$
    \item All the horizontal sticks have arrows pointing down
    \item Each horizontal stack is adjacent to four other horizontal stacks, two above and two below, so that the adjacency graph of the horizontal stacks forms an infinite (diagonal) grid
    \item Adjacent horizontal stacks overlap by at least one concave corner. That is, the adjacency between the top stick of a horizontal stack and the bottom stick of one of its neighbors is not as in case 2 in \Cref{fig:forbiddenpairs}
    \item Each cycle of four horizontal stacks making up a grid cell encloses a region that is filled by a vertical stack containing at least two sticks
    \item Each vertical stack splits into a left part and a right part, where the sticks in the left (resp. right) part have arrows pointing right (resp. left)
\end{itemize}

\begin{figure}
    \centering
    \includegraphics[page=4]{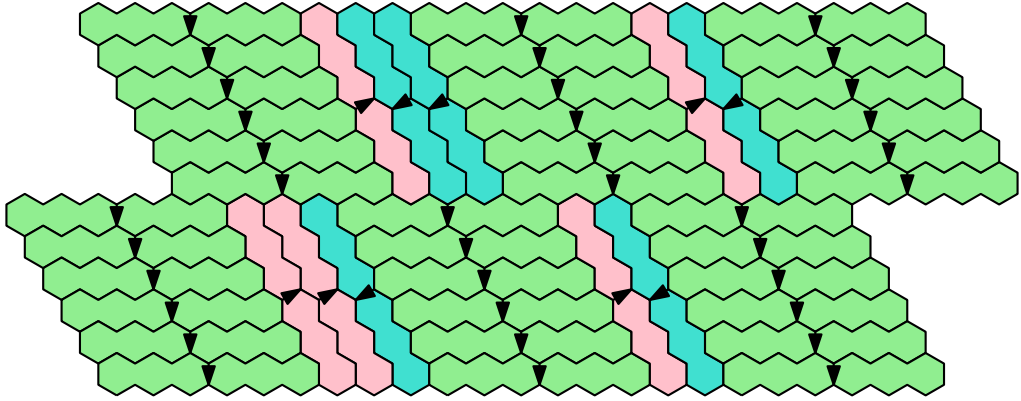}
    \caption{The \emph{weave pattern} that the sticks should form. }
    \label{fig:weavepattern}
\end{figure}

In order to create such a pattern, we forbid matchings between any of the following pairs of edges:

\begin{enumerate}
    \item $(y_i, y_j)$ for $i, j\in \{1, 2\}$
    \item $(x_i, x_j)$ or $(z_i, z_j)$ for $i, j \in \{1, 2\}$
    \item $(y_i, x_j)$ for $i, j\in \{1, 2\}$
    \item $(x_i, z_j)$ for $i, j\in \{1, 2\}$
    \item $(y_i, a_1)$ or $(y_i, c_1)$ for $i\in \{1, 2\}$
    \item $(x_1, c_i)$ for $1\le i\le n-1$
    \item $(z_2, d_i)$ for $1\le i\le n-2$
    \item $(x_2, a_i)$ for $1\le i\le n-1$
    \item $(z_1, b_i)$ for $1\le i \le n-1$
    \item $(y_i, z_2)$ for $i\in \{1, 2\}$
    \item $(y_i, d_{n-1})$ for $i\in \{1, 2\}$
\end{enumerate}

These are illustrated in \Cref{fig:forbiddenpairs}. Note that none of the forbidden cases occur in a weave pattern. It remains to check that forbidding these pairs is sufficient to force tiles to form the weave pattern. 

\begin{figure}
    \centering
    \includegraphics[page=19,scale=0.9]{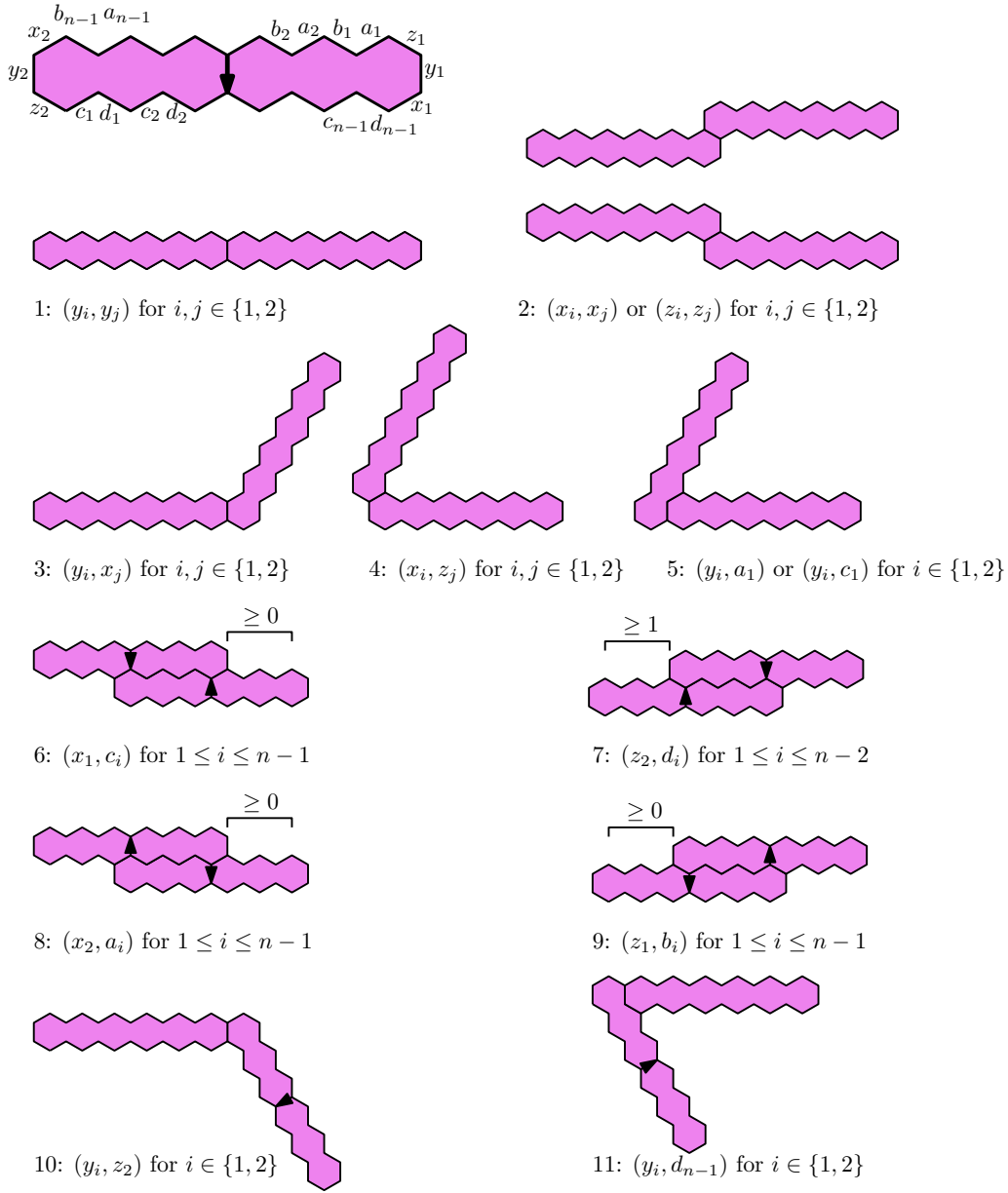}
    \caption{
    Preventing these arrangements of tiles will force a tiling to form a weave pattern. Note that these are \emph{not} symmetric under reflections. When arrows are not shown, both orientations are prevented.}
    \label{fig:forbiddenpairs}
\end{figure}

\begin{lemma}\label{lem:weavepattern}
Any tiling of the plane by a stick of length at least $5$ that satisfies the rules 1 through 11 must form a weave pattern.
\end{lemma}

\begin{proof}
Consider a tiling of the plane by sticks of length $n$. 

Consider a stack of sticks, WLOG oriented horizontally. For each stick in the stack, the spaces directly right and left of them must by filled by sticks that slant in the same direction that the stack leans. Otherwise, one of the rules would be violated, as shown in \Cref{fig:stackadjacent}. So a single stick can't be in both a left-leaning pair and a right-leaning pair, meaning that all the pairs in a single stack lean the same direction.

\begin{figure}
    \centering
    \includegraphics[page=6]{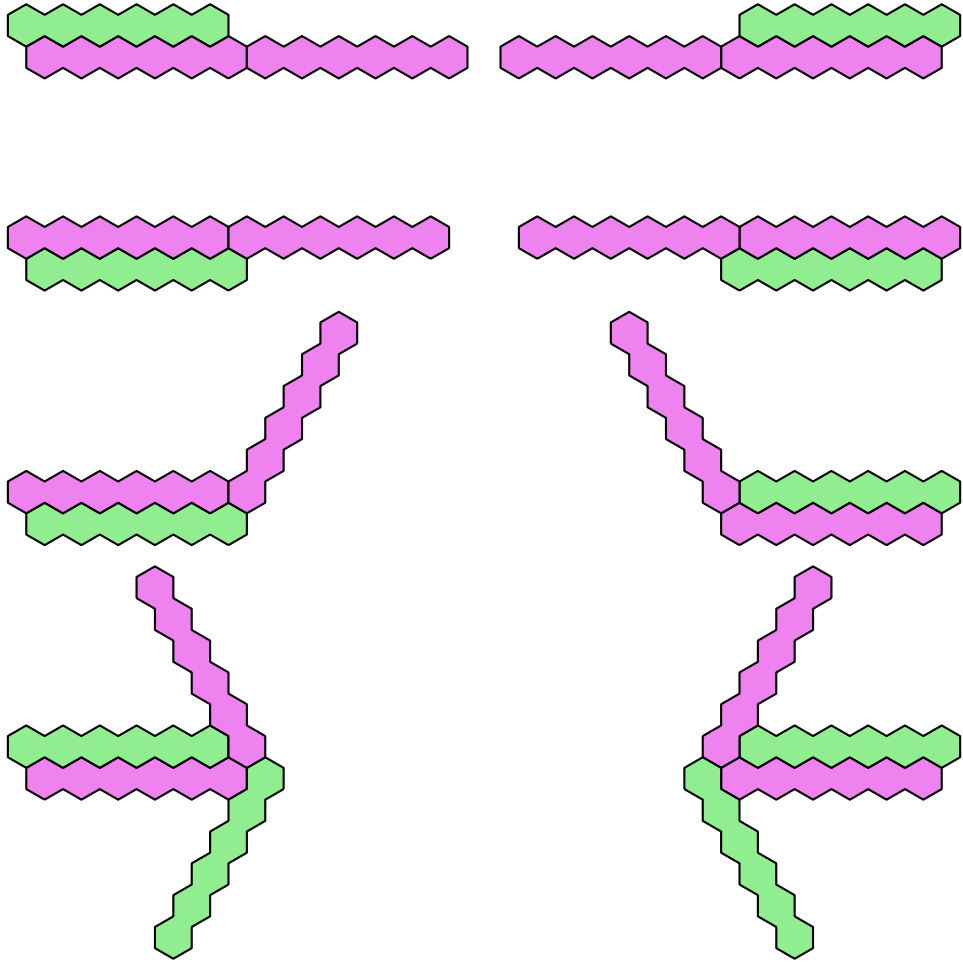}
    \caption{Trying to place a stick next to a stacked pair when the stick doesn't slant in the same direction that the stack leans. The pairs highlighted in purple violate one of the rules.}
    \label{fig:stackadjacent}
\end{figure}

If such a stack has height larger than $n$, then a stick next to that the stack violates rule 4, as shown in \Cref{fig:stackheightbound}. So each stack has height at most $n$. 

\begin{figure}
    \centering
    \includegraphics[page=7]{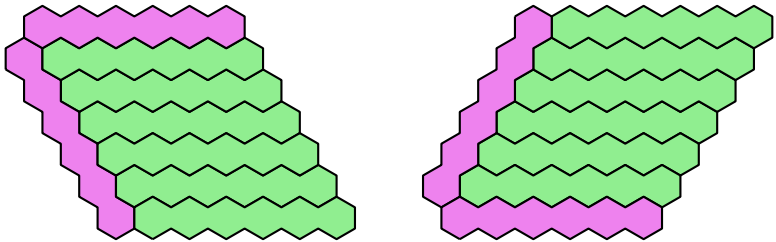}
    \caption{A stack of height more than $n$ introduces a pair violating rule 4 (highlighted in purple)}
    \label{fig:stackheightbound}
\end{figure}

If a left-leaning stack has height less than $n-1$, then this would introduce one of the cases shown in \Cref{fig:leftstacklowerbound1}, which are forbidden by the rules. If it has height exactly $n-1$, it must be as in one of the arrangements shown in \Cref{fig:leftstacklowerbound2}, which are also forbidden. So every left-leaning stack has height exactly $n$. Note that this is not true for right-leaning stacks.

\begin{figure}
    \centering
    \includegraphics[page=8]{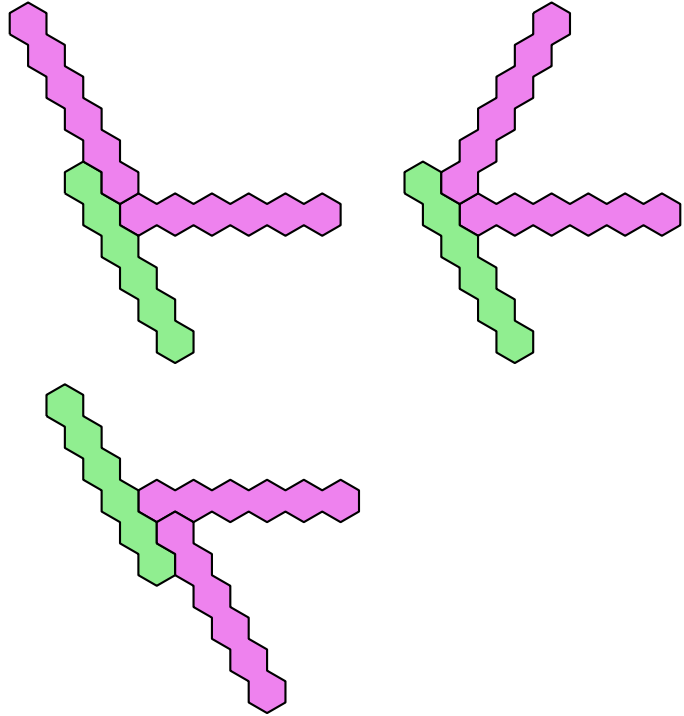}
    \caption{If a horizontal stick is placed next to a left-slanted stick, then all of the concave vertices on the right side the left-slanted stick must be filled by horizontal sticks. This shows that a left-leaning stack has height at least $n-1$.}
    \label{fig:leftstacklowerbound1}
\end{figure}

\begin{figure}
    \centering
    \includegraphics[page=9,width=\textwidth]{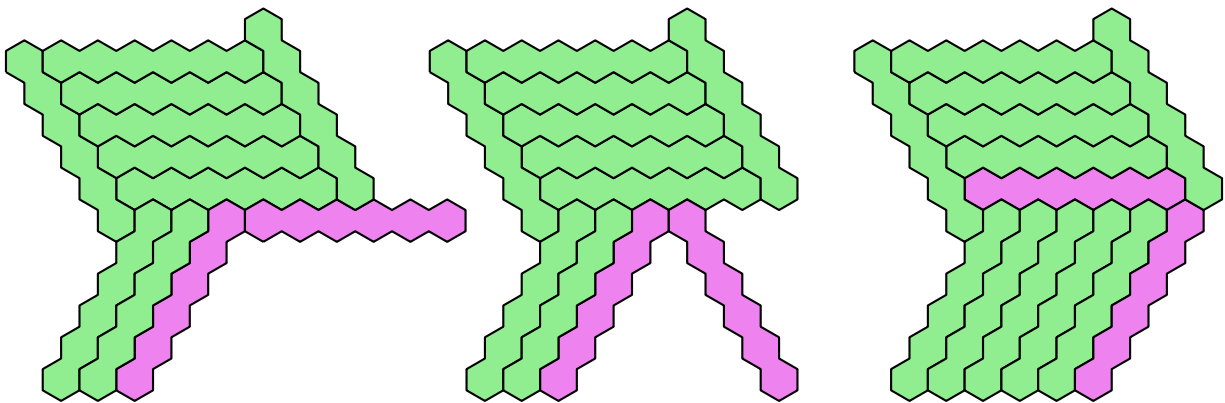}
    \caption{A left leaning stack can't have height exactly $n-1$, so it must have height $n$.}
    \label{fig:leftstacklowerbound2}
\end{figure}

Next, we argue that there must be a left leaning stack.
Since stacks have height at most $n$, there is a stick that is at the end of a stack or that is not part of a stack.
Call this stick $S$.
Orient $S$ so that it is horizontal and is at the bottom of a stack (if it is part of a stack).
Let $T$ be the stick that fills the space left of $S$, and note that $T$ is not horizontal.
If $T$ slants left, then $S$ is part of a left-leaning stack (otherwise it would introduce one of the configurations shown in \Cref{fig:leftstacklowerbound1}). If $T$ slants to the right, then the space below $S$ can only be filled by a left-leaning stack, as shown in \Cref{fig:alwayshavealeftstack}.

\begin{figure}
    \centering
    \includegraphics[page=10,width=\textwidth]{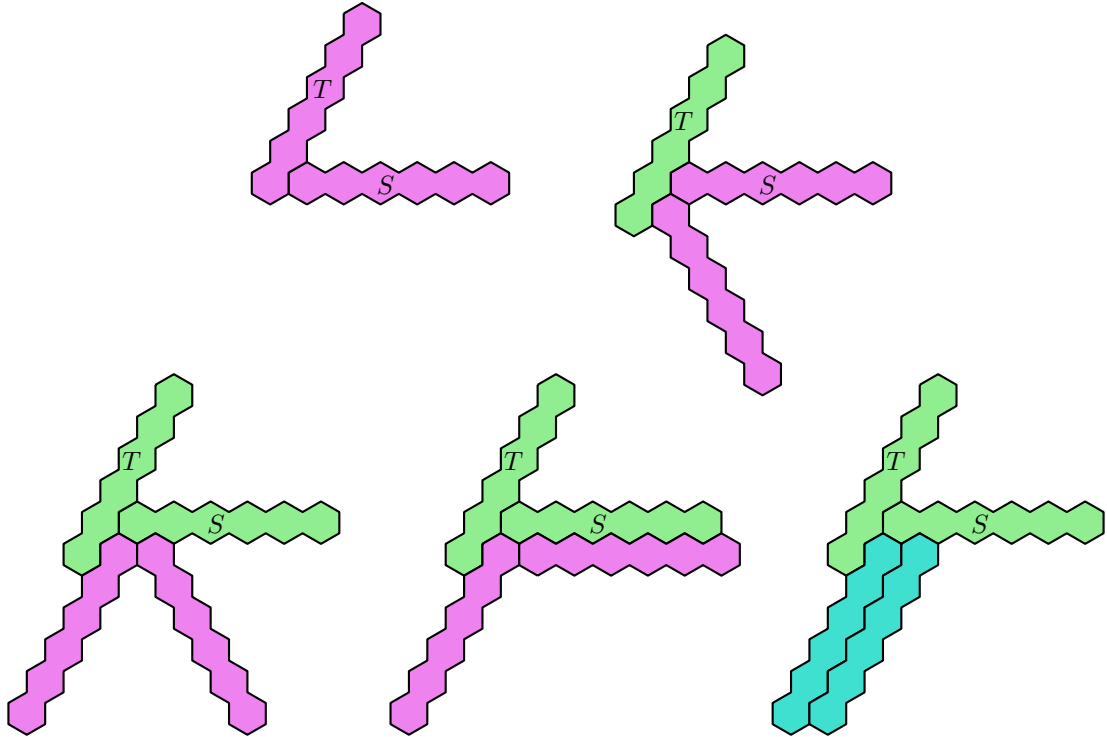}
    \caption{If $T$ slants right, then the space below $S$ must be filled by a left-leaning stack.}
    \label{fig:alwayshavealeftstack}
\end{figure}
 
Now choose some left leaning stack. By the above, it has length $n$. Rules 6 and 8 require that all the sticks in this stack have the same orientation. Orient the stack so that the sticks are horizontal with arrows facing down. The left-slanting sticks left and right of the stack have arrows pointing away from the stack because of rules 10 and 11. This configuration is shown in \Cref{fig:fullleftstack}. Let $S$ be the uppermost stick in this stack.

\begin{figure}
    \centering
    \includegraphics[page=11]{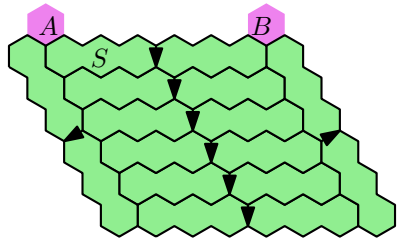}
    \caption{If $T$ slants right, then the space below $S$ must be filled by a left-leaning stack, as in the bottom-right figure.}
    \label{fig:fullleftstack}
\end{figure}

The space marked $A$ in \Cref{fig:fullleftstack} must be filled by a horizontal stick (call it $T$, see \Cref{fig:leftstackaboveleft1}) that meets some, but not all, of the concave corners on the back side of $S$; the other possibilities fail, as shown in \Cref{fig:fillingspaceA}. The space right of $T$ then must be filled by a left-slanted stick, otherwise introducing one of the cases shown in \Cref{fig:stickrightofT}. As argued above (see \Cref{fig:leftstacklowerbound1}), this means that $T$ is part of a left-leaning stack, which must be as in \Cref{fig:leftstackaboveleft2}. Note that $T$ has an arrow facing down because of rules 6 and 8.

\begin{figure}
    \centering
    \includegraphics[page=12]{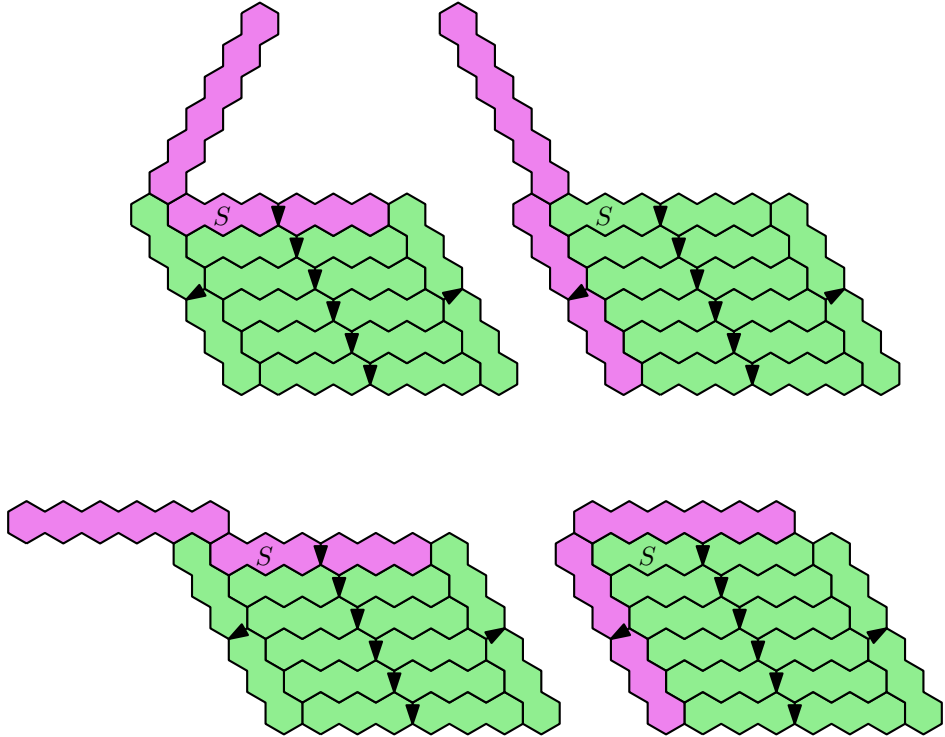}
    \caption{The space marked $A$ in \Cref{fig:fullleftstack} can't be filled in any of the ways shown.}
    \label{fig:fillingspaceA}
\end{figure}

\begin{figure}
    \centering
    \includegraphics[page=13]{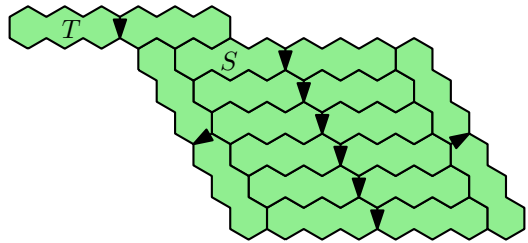}
    \caption{$T$ must be a horizontal stick, facing downwards, that meets some, but not all, of the concave corners on the back side of $S$.}
    \label{fig:leftstackaboveleft1}
\end{figure}

\begin{figure}
    \centering
    \includegraphics[page=14]{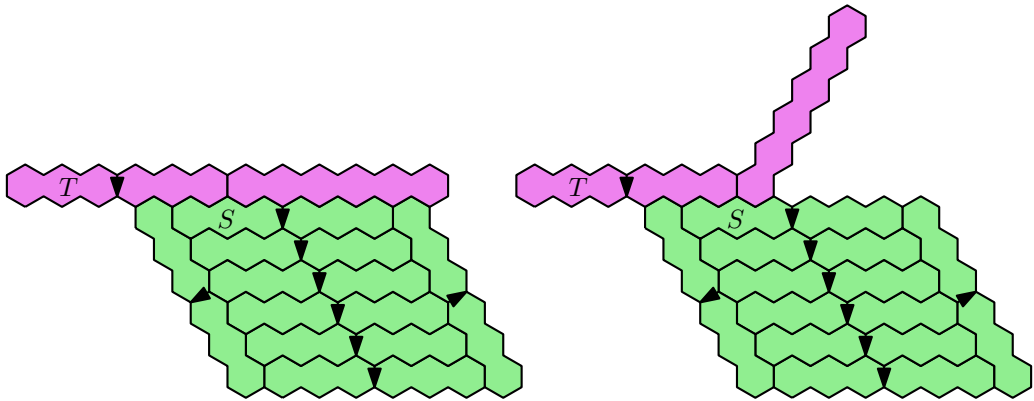}
    \caption{The space right of $T$ must contain a left-slanted stick.}
    \label{fig:stickrightofT}
\end{figure}

\begin{figure}
    \centering
    \includegraphics[page=15]{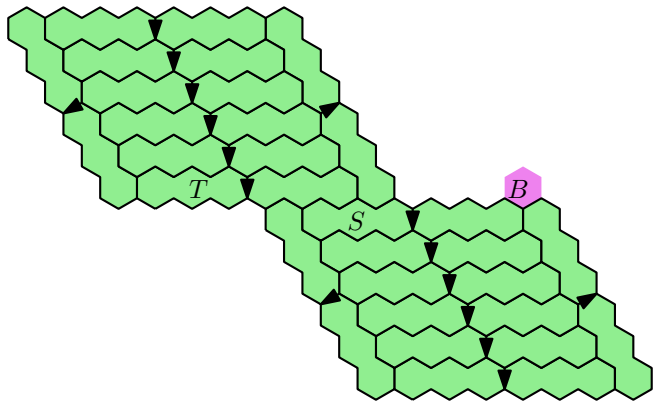}
    \caption{$T$ is the bottom of a left-leaning stack.}
    \label{fig:leftstackaboveleft2}
\end{figure}

Now we can argue that the space marked $B$ in \Cref{fig:fullleftstack,fig:leftstackaboveleft2} must be filled by a horizontal stick $U$ that meets some, but not all, of the concave corners on the back side of $S$. The space between $U$ and $T$ must be filled with left-slanted sticks; \Cref{fig:fillingspaceB} shows how other possible arrangements fail.
The stick $U$ must be oriented with the arrow pointing down by rules 7 and 9. Again, $U$ must be part of a left-leaning stack, shown in \Cref{fig:twostacksabove}.

\begin{figure}
    \centering
    \includegraphics[page=16,width=\textwidth]{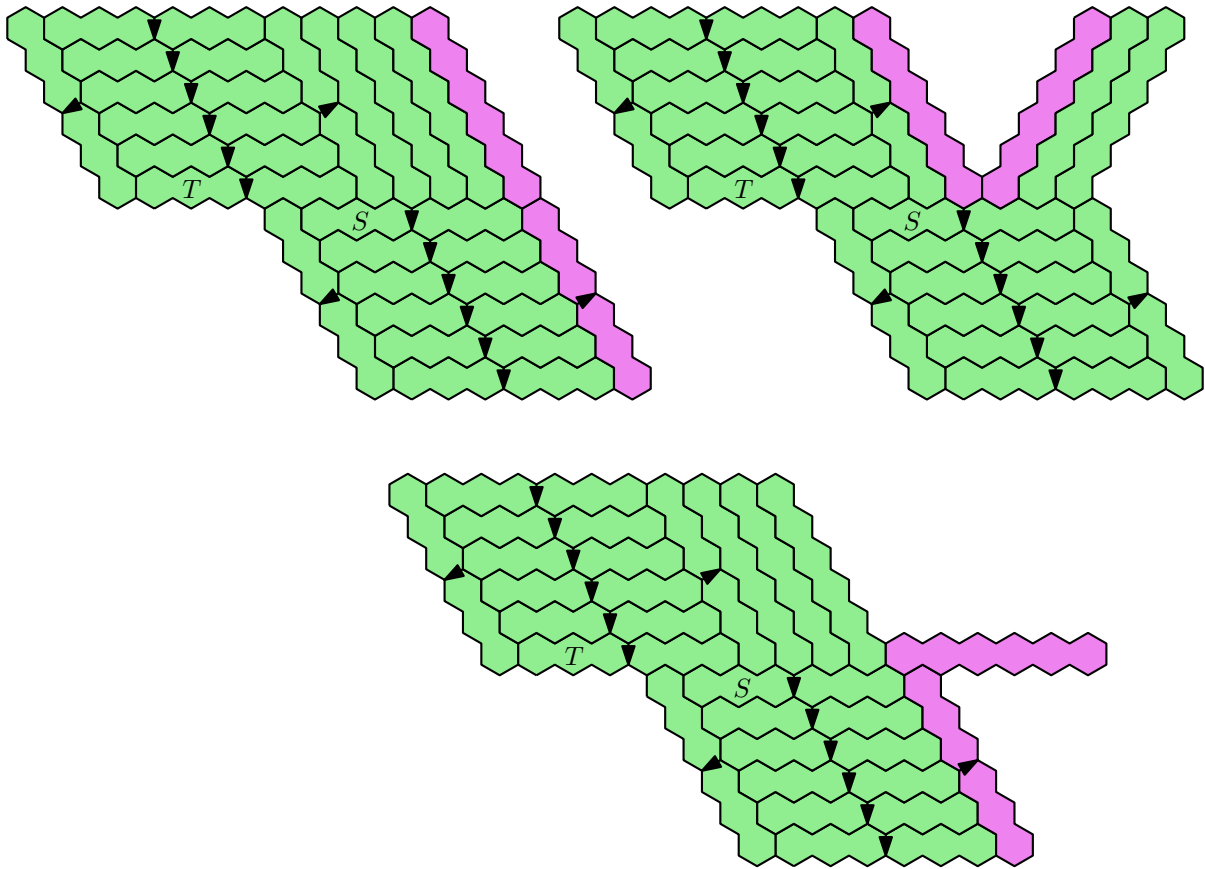}
    \caption{The space $B$ must also be filled by a horizontal stick meeting some, but not all, of the concave corners on the back side of $S$. We use the presence of the stack containing $T$ to make this conclusion.}
    \label{fig:fillingspaceB}
\end{figure}

\begin{figure}
    \centering
    \includegraphics[page=17]{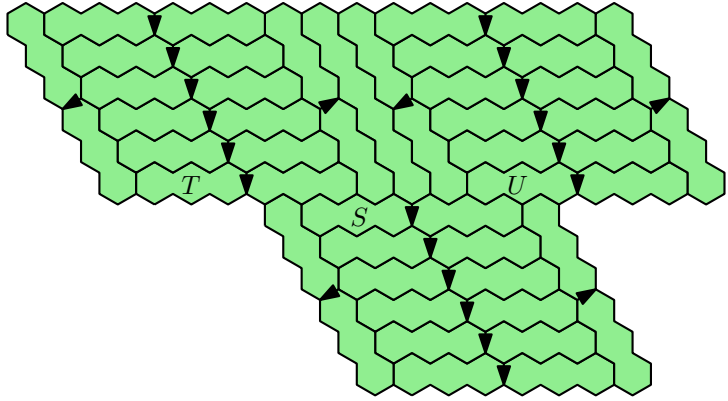}
    \caption{Every left-leaning stack has two left-leaning stacks meeting its back side.}
    \label{fig:twostacksabove}
\end{figure}

Examining the right-leaning stack of left-slanted sticks between $T$ and $U$, we see that the leftmost of these sticks has an arrow pointing to the right and the rightmost has an arrow pointing to the left. Two sticks in a right leaning stack can be stacked front-to-front, but not back-to-back (because of rule 9). So some number of sticks on the left must face right and the rest of them face left. The number of sticks between $T$ and $U$ is not fixed, but is at least $2$ and at most $n-3$.

A similar argument that every left-leaning stack also has meets two left-leaning stacks below, as in \Cref{fig:fourstacksaround}. It remains to check that the left-leaning stacks forms a grid. Specifically, if $\mathcal{S}$ is a left-leaning stack, we need to check that the stack above-right of the stack above-left of $\mathcal{S}$ is the same as the stack above-left of the stack above-right of $\mathcal{S}$. But this is straightforward.

\begin{figure}
    \centering
    \includegraphics[page=18]{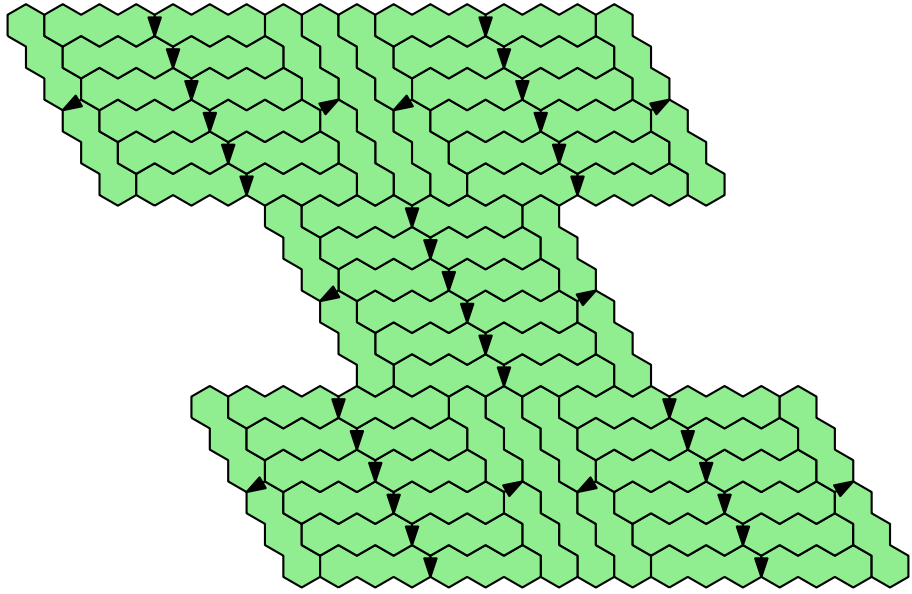}
    \caption{Every left leaning stack has $4$ adjacent left-leaning stacks---two on the front side and two on the back side.}
    \label{fig:fourstacksaround}
\end{figure}

By induction, the entire plane must be filled with a grid of left-leaning stacks meeting each other in this way. So the tiling forms a weave pattern.
\end{proof}

\subsection{AB tiles}

Given a set of Wang tiles, we want to construct a stick tile and some edge-to-edge matching rules so tilings with the stick tile correspond to tilings with the Wang tiles. Instead of simulating Wang tiles directly, we go through an intermediate problem that we call \emph{AB tiling}.

In the AB tiling problem, we are given two sets $A$ and $B$ of $1\times 2$ rectangular tiles, where $|A|=|B|$. The left and right long edges of each tile are each assigned a top color and a bottom color. The problem asks whether the tiles can tile the plane without rotation, where an $A$ tile must be placed so that its lower-left corner has even $x$ and $y$ coordinates, and a $B$ tile must be placed so that its lower-left corner has odd $x$ and $y$ coordinates (that is, they form a pattern like the one shown in \Cref{fig:abpattern}). The colors should match along each vertical edge.

\begin{figure}
    \centering
    \includegraphics[page=1]{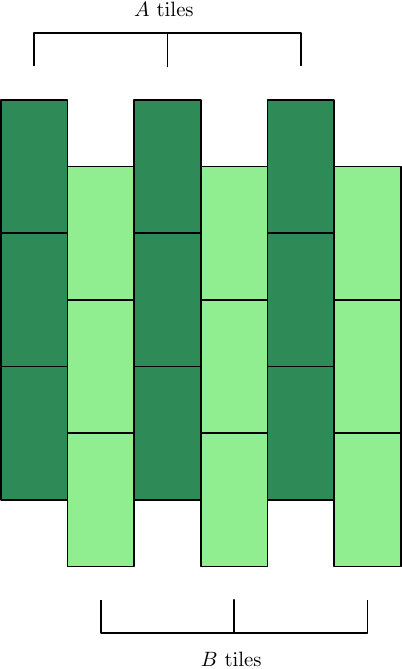}
    \caption{The pattern that an AB tiling should make.}
    \label{fig:abpattern}
\end{figure}

\begin{lemma}\label{lem:ABtiling}
AB tiling is undecidable. 
\end{lemma}

\begin{proof}
We reduce from Wang tiling; the construction is illustrated in \Cref{fig:absimulation}.
Consider a set of $n$ Wang tiles with colors $\{1, \dots, k\}$.
We will create a set of AB tiles with $|A|=|B|=n$ and colors $\{0, \dots, n+k\}$.

For each Wang tile, we create one $A$ tile and one $B$ tile. Suppose the $i$th Wang tile has colors $a, b, c$ and $d$ on the left, right, upper, and lower faces respectively. Then the corresponding $A$ tile has colors $0$, $c$, $a$ and $k+i$ on the upper left, upper right, lower left, and lower right faces respectively, and the $B$ tile has colors $k+i$, $b$, $d$, $0$ on those faces.

For each Wang tile we can make an S-tetromino made from the corresponding $A$ and $B$ tiles. Every tile must be part of one of these tetrominoes since the $A$ and $B$ tiles corresponding to the $i$th Wang tiles are the only tiles with color $k+i$. It is easy to see that tilings with these tetrominos are in correspondence with tilings by the Wang tiles.

\begin{figure}
    \centering
    \includegraphics[page=2]{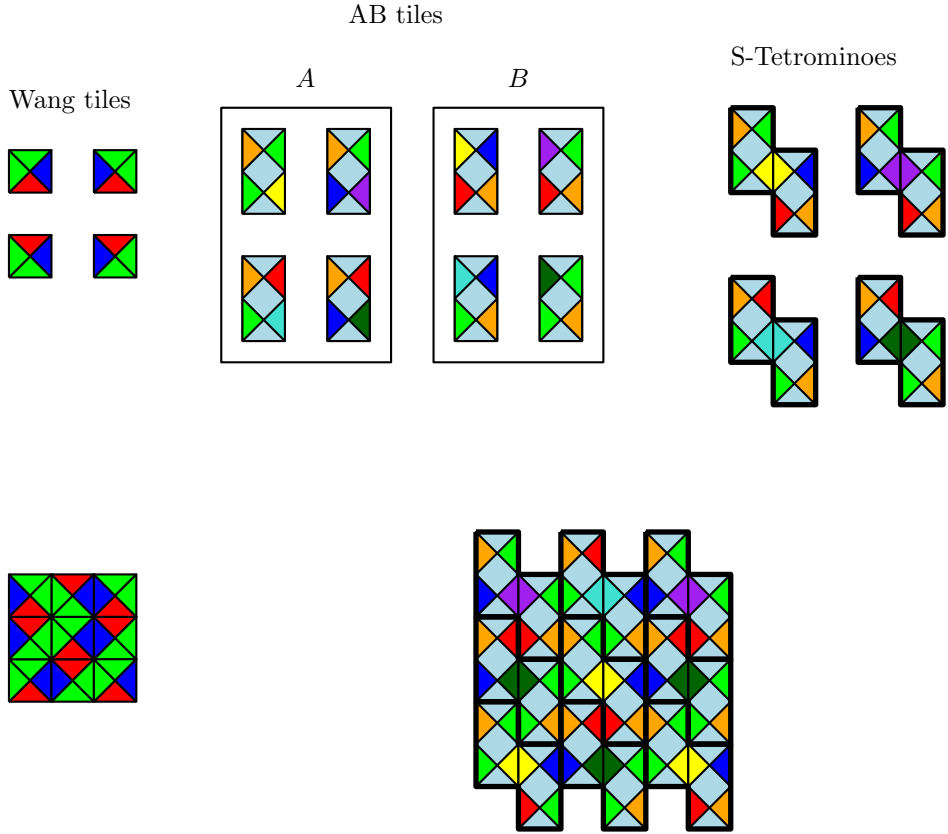}
    \caption{A set of Wang tiles (upper left) can be simulated by a set of AB tiles (upper middle). The tiles join to form S-tetrominoes (upper right). The lower left figure shows a patch of a tiling with these Wang tiles and the lower right shows a simulation of this patch with one S-tetromino per Wang tile.}
    \label{fig:absimulation}
\end{figure}
\end{proof}

\subsection{Tiling schematics}\label{sec:schematics}

We represent weave patters by simplified schematics. An example is shown in \Cref{fig:tilingschematic}.

\begin{figure}
    \centering
    \includegraphics[page=1]{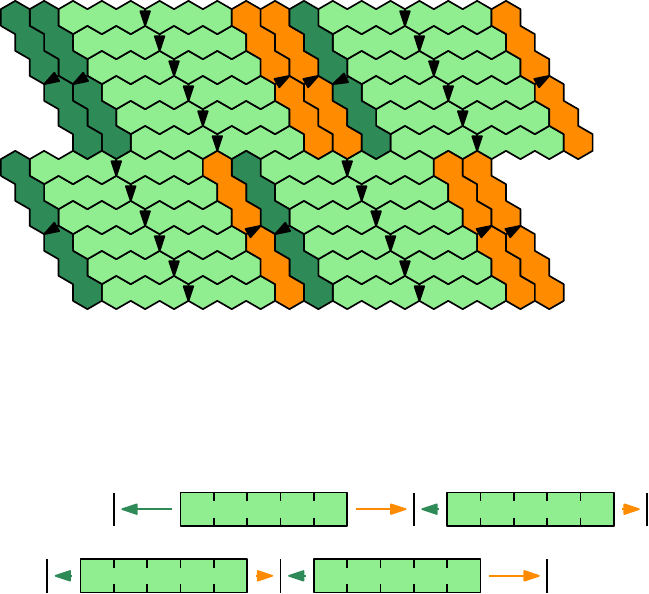}
    \caption{A section of a weave pattern tiling together with a schematic representation. Two concave vertices in the top tile of the lower-left stack are filled by the bottom tile in the upper-left stack, so the corresponding blocks overlap by $2$ units. }
    \label{fig:tilingschematic}
\end{figure}

A \emph{tiling schematic} consists of rows of $n\times 1$ blocks. Between adjacent blocks in a row there are gaps of length at least $2$, and each gap is dividing into a left part and right part, each of length at least $1$.
The left part of a gap is represented by a right-facing arrow and the right part of a gap is represented by a left-facing arrow. The \emph{left length} (resp. \emph{right length}) of a gap is the length of the left (resp. right) part of the gap.
The rows are put together so that the blocks in adjacent rows overlap by at least one unit.

A schematic with blocks of length $n$ corresponds to a weave-pattern tiling of sticks with length $n+1$. Every block in the schematic represents a horizontal stack in the tiling.
The top edge of a block is divided into $n$ unit segments, corresponding to the $n$ concave vertices along the upper side of the top tile in the corresponding horizontal stack. Similarly, the unit segments on the bottom edge of a block correspond to the concave vertices along the lower side of the bottom tile in that stack. 
We define \emph{position} $i$ on a block as the point at distance $1+i$ from the left side of the block.

In order to simulate an AB tiling, we need to add some edge-to-edge matching rules in addition to those shown in \Cref{fig:forbiddenpairs}. These will be represented in the schematics by arrow and triangle marks placed along the top and bottom faces of a block. Since we want to construct just a single stick tile, we need to place the same markings on every block in the schematic. The markings can represent the following types of constraints:

\begin{itemize}
    \item An arrow marking on top (resp. bottom) of a block should line up with the arrows in a gap if it occurs below (resp. above) a gap. The forbidden cases are shown in \Cref{fig:arrowmarking}. A left (resp. right) arrow mark on the top of a block pointing at position $i$ means that the pair $(a_{n-1-i}, y_2)$ (resp. $(a_{n-i}, y_1)$) is forbidden. A left (reps. right) arrow mark on the bottom of a block pointing at position $i$ means that the pair $(c_{i+1}, y_1)$ (resp. $(c_{i+2}, y_2)$) is forbidden (note that $n$ here is the length of the block, which is $1$ less than the number of hexagons in the stick).
    \item Triangles on the top of a block represent the places where the ends of gaps in the upper row can be. The allowed and forbidden cases are shown in \Cref{fig:trianglemarking}. The pair $(a_{n-1-i}, x_1)$ is forbidden unless there is a triangle mark facing position $i$, and the pair $(b_{n-i}, z_2)$ is forbidden unless there is a triangle facing position $i$.
\end{itemize}

\begin{figure}
    \centering
    \includegraphics[page=2]{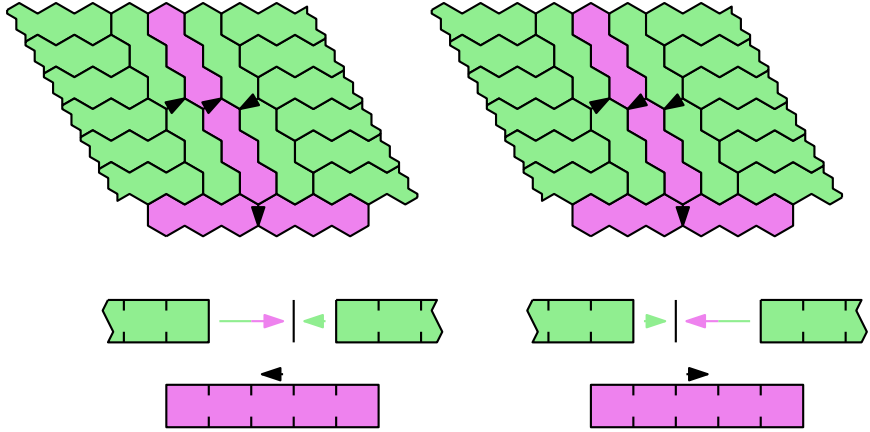}
    \caption{The forbidden configurations represented by an arrow marking.}
    \label{fig:arrowmarking}
\end{figure}

\begin{figure}
    \centering
    \includegraphics[page=3]{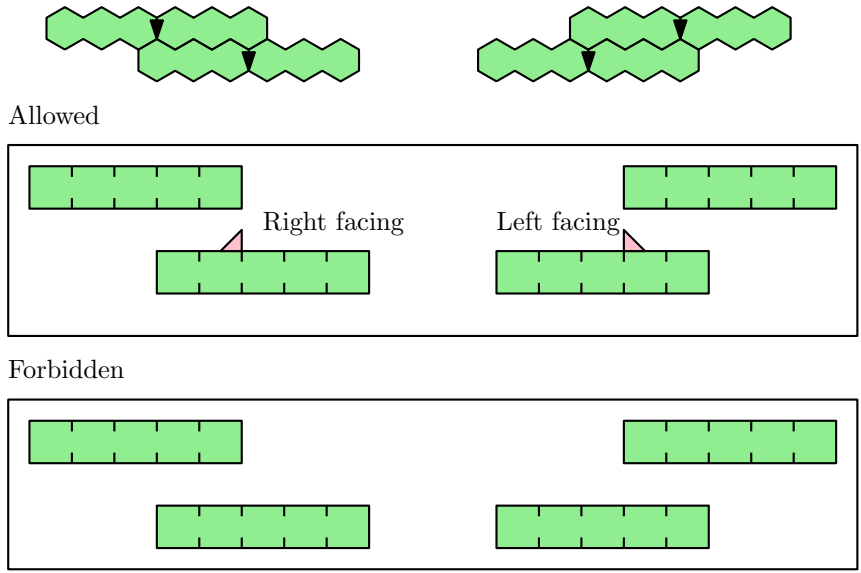}
    \caption{The ends of a block in the row above can only occur above a triangle marking facing in that direction. }
    \label{fig:trianglemarking}
\end{figure}

These types of edge pairs can only occur in a weave pattern in the ways represented by \Cref{fig:arrowmarking,fig:trianglemarking}. So it is clear that a schematic satisfying these conditions with respect to the markings will yield a stick tiling that satisfies the corresponding matching rules, and vice versa.

\subsection{States and buckets}

\Cref{fig:simpleoperation} shows an example of the type of constraint that can be created by combining arrow and triangle markings. In the example, the left (resp. right) length of a gap differs from the right (resp. left) length of the gap below it to the left (resp. right) by at most $1$. The triangle marks prevent gaps in the lower row from having lengths other than $1$ or $3$ (it is possible for them to both be $1$ or both be $3$, which is not shown in the figure). 

\begin{figure}
    \centering
    \includegraphics[page=4]{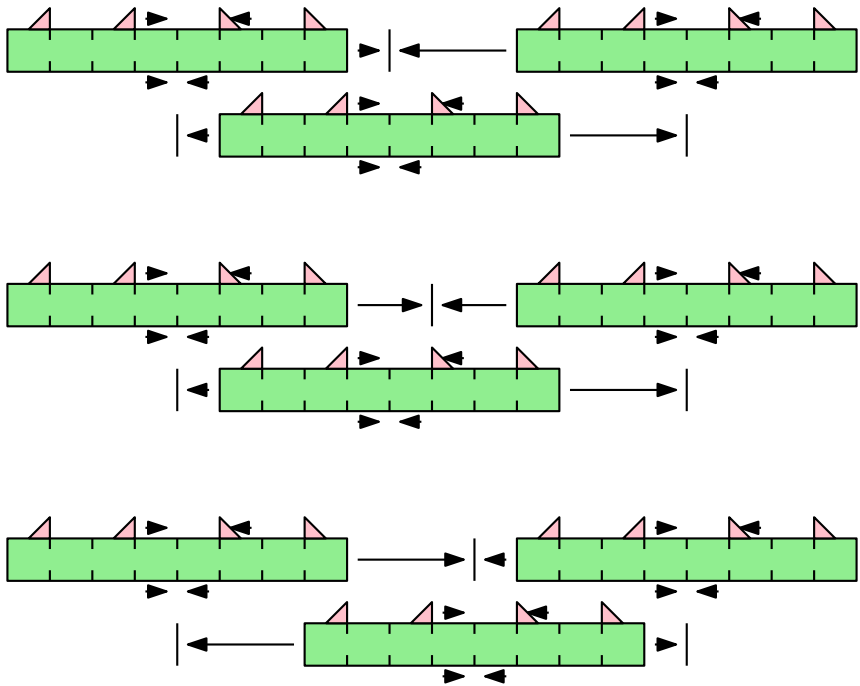}
    \caption{An example of the type of constraint that can be created with arrow and triangle markings. Triangle marks on top and arrow marks on the bottom of the block are used to restrict the set of possible lengths for the left and right sides of the gaps. Triangle marks on top of the block create a constraint between gaps in adjacent rows. }
    \label{fig:simpleoperation}
\end{figure}

We will enforce the structure of an AB tiling with a sequence of constraints like this. Different constraints occur depending on what part of the block is below a gap in the row above it. We divide the block into \emph{buckets}, as in \Cref{fig:morebuckets}. Each bucket represents some simple constraint. 

\begin{figure}
    \centering
    \includegraphics[page=6]{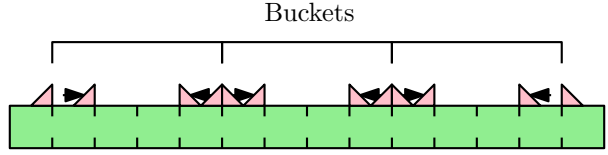}
    \caption{The block is divided into a number of equally-sized buckets. The boundary between adjacent buckets has a pair of arrows facing back-to-back. The left (resp. right) half of a bucket can have right (resp. left) triangle marks. All possible triangle marks in this example are shown.}
    \label{fig:morebuckets}
\end{figure}

The boundary between adjacent buckets is marked on the top of the block with a pair of arrows facing back-to-back (as in \Cref{fig:morebuckets}). Two adjacent spaces in a gap have directions facing front-to-back or front-to-front, so gaps can't cross these boundaries. This means that each gap is above a single bucket.

If we place the triangle marks carefully, then we can ensure that all the gaps in a single row are above the same bucket in the row below.
We say that a row $R$ is in state $i$ if each gap in the row above $R$ is above the $i$th bucket of some block in row $R$.

The states define a sequence of operations that should repeat in a cyclic pattern. If a row has state $i$, then the row below it should have state $\sigma(i)$, where $\sigma$ is a cyclic permutation. To enforce this, we can control the left and right lengths of each gap, so that the gaps in a row with state $i$ can only fit above bucket $\sigma(i)$. A simple example is shown in \Cref{fig:statecycle}. 

\begin{figure}
    \centering
    \includegraphics[page=7,width=\textwidth]{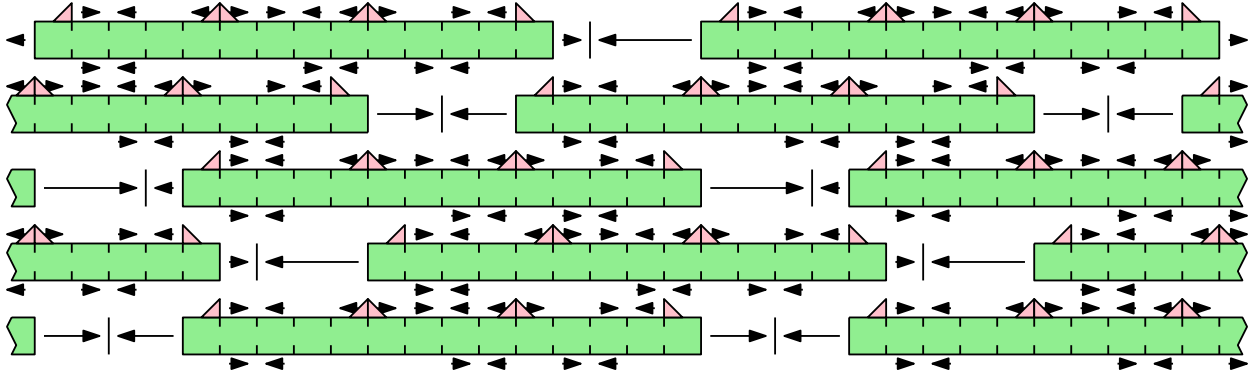}
    \caption{A block with $3$ buckets. Bucket $i$ can only fit below a gap with left length $i$ and right length $4-i$. Each row has gaps that only fit above the next bucket in the cycle, so the rows cycle through the states. }
    \label{fig:statecycle}
\end{figure}

In what follows, we give a detailed construction of the buckets. We will ensure that the rows cycle through the states in a particular order, and we will describe in detail the types of constraint that a single state can represent. The goal is to specify the buckets in such a way that the set of valid tilings is fully characterized by a simple set of conditions called the \emph{gap conditions}.

\subsection{The gap conditions}\label{sec:gapconditions}

The block is divided into $s$ buckets as in \Cref{fig:morebuckets}. These are labeled $1, \dots, s$. Let $\ell_b$ be the length of a bucket, so bucket $i$ spans from positions $(i-1)\ell_b$ to $i\ell_b$. For each $1\le i\le s$, there is a right arrow mark facing position $(i-1)\ell_b+1$ and a left arrow mark facing position $i\ell_b-1$. 

We place right-facing triangles only at positions $x$ satisfying $(i-1)\ell_b\le x < (i-\frac12)\ell_b$ for some $i$. We place left-facing triangles only at positions $x$ satisfying $(i-\frac12)\ell_b< x\le i\ell_b$. That is to say, only the left half of a bucket can have right-facing triangle marks and only the right half of a bucket can have left-facing triangle marks. All possible triangle marks are shown in \Cref{fig:morebuckets}.

\begin{lemma}\label{lem:gapsinsamebucket}
Each gap is contained in one of the buckets in the row below. If two gaps are in the same row, then they occur above buckets with the same index.
\end{lemma}

\begin{proof}
The rightmost square of bucket $i$ has a left-facing arrow while the leftmost square of bucket $i+1$ has a right-facing arrow. These arrows are facing back-to-back, but adjacent spaces in a gap above either have arrows facing the same direction or arrows facing front-to-front. So one side of this boundary must be underneath a block and not underneath a gap. The two spaces at the ends of the block are not in any bucket, but the blocks must overlap by at least one unit, so these spaces can't be below a gap.
This shows that the gap is contained in a bucket in the row below.

Now we argue that two gaps in the same row must be above buckets with the same index. The argument is illustrated by \Cref{fig:gapsinsamebucket}.
Suppose that the left end of a block $X$ in an upper row occurs above position $(i-1)\ell_b+j$ in a block $Y$ below it, with $1\le j\le \ell_b$. There must be a left-facing triangle mark at this position, so $\frac12\ell_b < j\le \ell_b $. If $Z$ is the block right of $Y$, then the right end of $X$ occurs above position $(i-1)\ell_b+j-k$ in $Z$, where $k$ is the size of the gap between $Y$ and $Z$. Note that  $2\le k\le \ell_b$, so $-\frac12\ell_b<j-k<\ell_b$.

There is a right-facing triangle mark below the right end of $X$, so $(i'-1)\ell_b\le (i-1)\ell_b+j-k<(i'-\frac12)\ell_b$ for some $i'$. So $(i-i')\ell_b\le j-k<(i'-i+\frac12)\ell_b$. Since $-\frac12\ell_b<j-k<\ell_b$, we have that $-1<i-i'<1$, and conclude that $i=i'$.

So any two adjacent gaps in the same row are above the same index bucket in the row below. By induction, this is true for any pair of gaps in that row.
\end{proof}

\begin{figure}
    \centering
    \includegraphics[page=8,width=\textwidth]{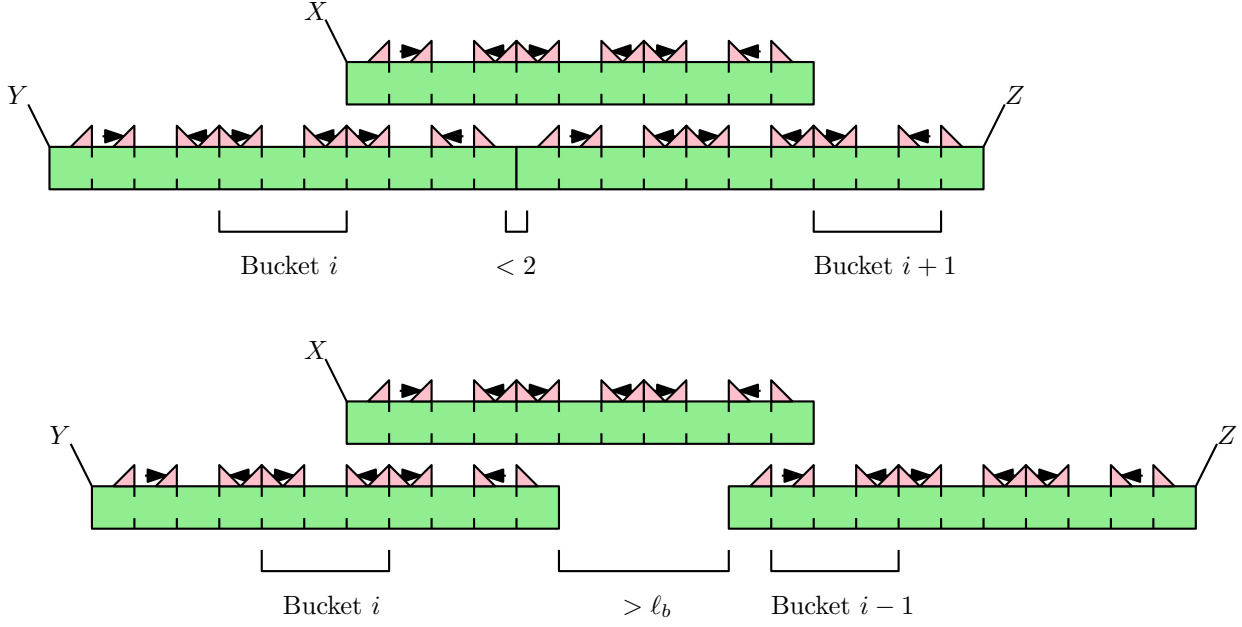}
    \caption{If two adjacent gaps in a row occur above different buckets in the row below, then there is a gap in that row which is either too long or too short.}
    \label{fig:gapsinsamebucket}
\end{figure}

Recall that a row $R$ is in state $i$ if each gap in the row above $R$ is above the $i$th bucket of some block in row $R$.
\Cref{lem:gapsinsamebucket} says that we can always assign a state to a row in this way.
We want the rows to progress through the states in a cycle.
Let $\sigma(i)=i+1$ for $1\le i\le s-1$, and $\sigma(s)=1$. If a row has state $i$, then we want to make sure that the row below it will have state $\sigma(i)$ (see \Cref{lem:statecycle}).

We now specify each bucket. We say that relative position $k$ in bucket $i$ is the position a distance $k$ from the start of the bucket (that is, position $(i-1)\ell_b+k$). 

The left and right lengths of a gap need to encode three different numbers. The first is a \emph{state}, which should be the same across gaps in a row. The other two numbers are called the left and right \emph{values} of the gap. We fix a number $v$ of allowed values, and label the values $1, \dots, v$. We say that a gap in a row with state $i$ has left value $1\le x\le v$ and right value $1\le y\le v$ if it has left length $(2v-1)(\sigma(i)-1)+x$ and right length $(2v-1)(s-\sigma(i))+y$. We will construct our tiles so that every gap in a valid tiling has left and right lengths of this form (see \Cref{lem:leftandrightsidelengths}). Since $2v-1\ge v$, the state $i$ and left and right values $x$ and $y$ can be uniquely recovered from left and right gaps lengths of this form.

A gap with left value $x$ and right value $y$ has total length $(2v-1)(s-1)+x+y$ (this does not depend on the state). Since $x$ and $y$ take values in $1, \dots, v$, this is at least $(2v-1)(s-1)+2$ and at most $(2v-1)(s-1)+2v$. We set $\ell_b=(2v-1)(s-1)+2v$, which is the length of the longest gap that we want to allow.

The markings in the $i$th bucket are determined by $3$ parameters $L_i$, $R_i$ and $I_i$.
The set $L_i\subseteq\{1, \dots, v\}$ is the set of allowed left values and $R_i\subseteq \{1, \dots, v\}$ is the set of allowed right values.
That is, every gap in a row with state $i$ should have a left value in $L_i$ and a right value $R_i$ (see \Cref{lem:leftandrightsidelengths}).
The set $I_i\subseteq \{-v+1,\dots, v-1\}$ is a discrete interval that bounds the difference in values between adjacent states (see \Cref{lem:intervalbound}). 
We require that $0\in I_i$. 

\Cref{fig:bucketschematic} illustrates how $L_i, R_i$ and $I_i$ are used to place markings in the $i$th bucket.
There are right-facing triangle marks at relative positions $v-y$ for $y\in R_i$ and left-facing triangle marks at relative positions $\ell_b-v+x$ for $x\in L_i$.
There is also a right arrow mark pointing at $(2v-1)(i-1)+v-\max(I_i)$ and a left arrow mark pointing at $(2v-1)(i-1)+v-\min(I_i)$.

\begin{figure}
    \centering
    \includegraphics[page=9]{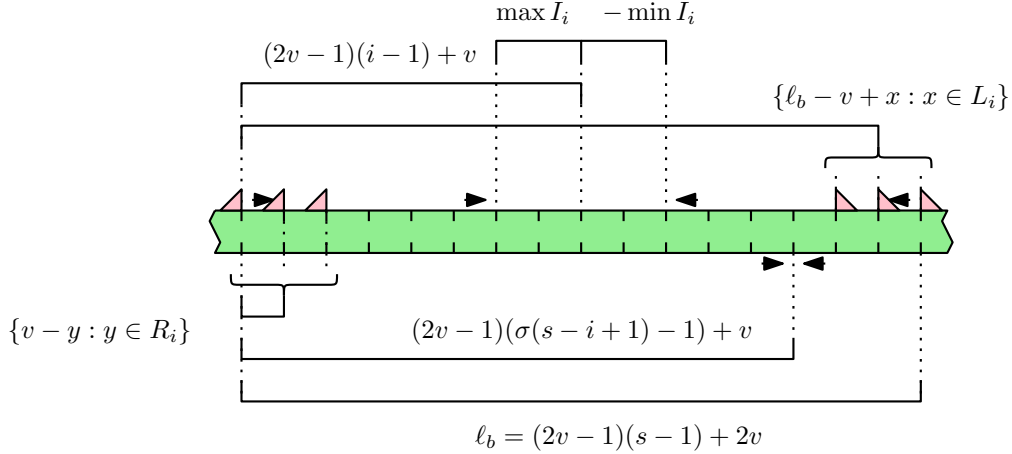}
    \caption{Specification of a bucket. The specific example illustrated is the $i$th bucket where $s=3$, $v=3$, $i=2$, $L_i=R_i=\{1, 2, 3\}$ and $I_i=\{-2, 1, 0, 1, 2\}$. The arrow marks below bucket $i$ are for state $s-i+1$.}
    \label{fig:bucketschematic}
\end{figure}

We also need some marks on the bottom side of the block. Below bucket $i$, we place a pair of left and right arrows facing relative position $(2v-1)(\sigma(s-i+1)-1)+v$ as in \Cref{fig:bucketschematic}. Note that if the right (resp. left) end of a block $X$ is above bucket $i$ on $Y$, then the left (resp. right) end of $Y$ is below bucket $s-i+1$ on $X$. So the marks below bucket $i$ actually correspond to state $s-i+1$. 

\begin{lemma}\label{lem:leftandrightsidelengths}
A gap with state $i$ has left length $(2v-1)(\sigma(i)-1)+x$ and right length $(2v-1)(s-\sigma(i))+y$ for some $x\in L_i$ and $y\in R_i$. 
\end{lemma}

\begin{proof}
Consider a block $X$, and let $Y$ and $Z$ be the blocks below it to the left and right respectively, as in \Cref{fig:leftandrightsidelengths}.
The gaps in the row containing $X$ occur above some bucket $i$, where $i$ is the state of the row containing $Y$ and $Z$.
So the right end of $X$ is above position $(i-1)\ell_b+v-y$ on $Z$ for some $y\in R_i$ (because these are the positions of right-facing triangle marks), and the left end of $X$ is above position $i\ell_b-v+x$ on $Y$ for some $x\in L_i$.
So the right end of $Y$ is below position $(s-i)\ell_b+v-x$ on $X$, and the left end of $Z$ is below position $(s-i+1)\ell_b-v+y$ on $X$ (both of these are in bucket $s-i+1$ of $X$).
Between these is a pair of right and left arrow marks facing position $(s-i)\ell_b+(2v-1)(\sigma(s-(s-i+1)+1)-1)+v=(s-i)\ell_b+(2v-1)(\sigma(i)-1)+v$, so the direction change in the gap between $Y$ and $Z$ needs to occur below this position.
We calculate that this gap has left length $(s-i)\ell_b+(2v-1)(\sigma(i)-1)+v-((s-i)\ell_b+v-x)=(2v-1)(\sigma(i)-1)+x$ and right length $(s-i+1)\ell_b-v+y-((s-i)\ell_b+(2v-1)(\sigma(i)-1)+v)=\ell_b+y-(2v-1)(\sigma(i)-1)-2v=(2v-1)(s-1)+2v+y-(2v-1)(\sigma(i)-1)-2v=(2v-1)(s-\sigma(i))+y$.
\end{proof}

\begin{figure}
    \centering
    \includegraphics[page=10,width=\textwidth]{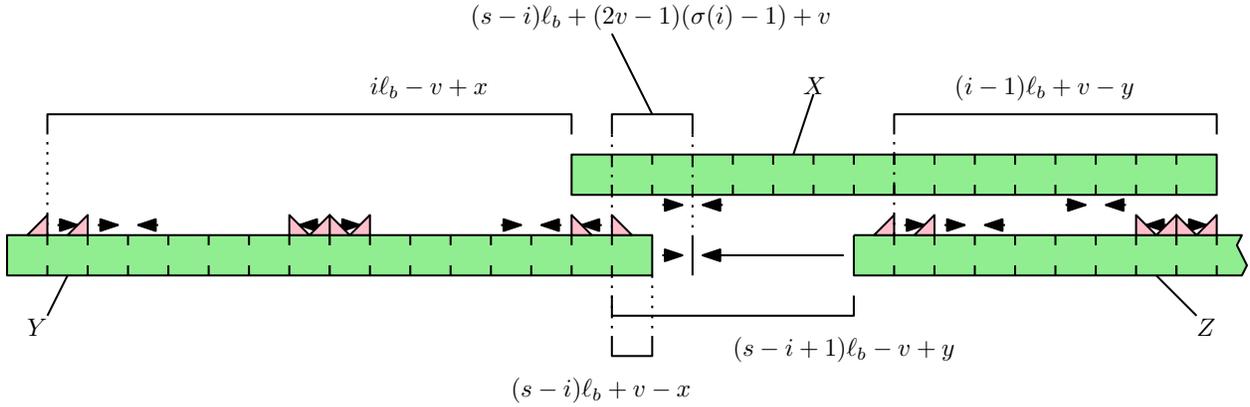}
    \caption{Calculating the possible left and right lengths for a gap in state $i$. }
    \label{fig:leftandrightsidelengths}
\end{figure}

The construction of the buckets ensures that a gap with left length $(2v-1)(i-1)+x$ and right length $(2v-1)(s-i)+y$ can only fit above bucket $i$. By the above, this means that if a row is in state $i$, then the row below it is in state $\sigma(i)$. An example is shown in \Cref{fig:bucketgapcompatability}.

\begin{figure}
    \centering
    \includegraphics[page=11]{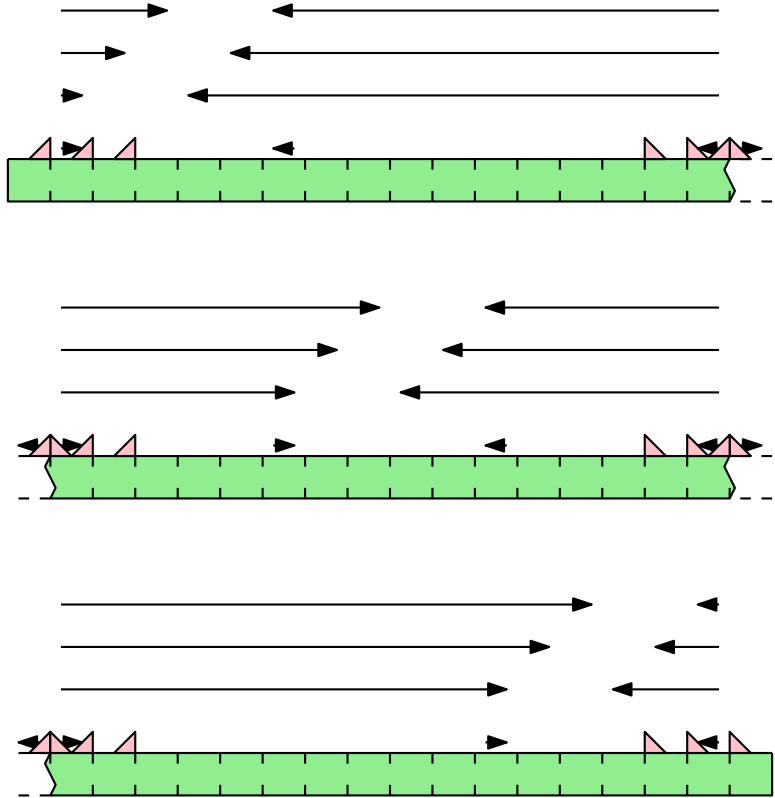}
    \caption{A stick with $s=3$ and $v=3$, split into three rows with one bucket on each row. Each $L_i$, $R_i$, and $I_i$ is chosen to be as large as possible, that is $L_i=R_i=\{1, 2, 3\}$ and $I_i=\{-2, -1, 0, 1, 2\}$. Above each bucket are drawn arrows representing the possible left and right gap lengths that can occur above that bucket. The left lengths that can occur have form $(2v-1)(i-1)+x$ for some $1\le i\le s$ and $1\le x\le v$, and the right lengths that can occur have form $(2v-1)(s-i)+y$ for some $1\le i\le s$ and $1\le y\le v$. Each of these lengths can only occur above one of the buckets.}
    \label{fig:bucketgapcompatability} 
\end{figure}

\begin{lemma}\label{lem:statecycle}
If a row is in state $i$, then the row below is in state $\sigma(i)$. 
\end{lemma}

\begin{proof}
First, calculate the possible left lengths of a gap above bucket $j$. The left end of such a gap is at relative position $v-y$ for some $y\in R_i$.
The bucket has a right arrow mark at relative position $(2v-1)(j-1)+v-\max(I_j)$ and a left arrow mark at relative position $(2v-1)(j-1)+v-\min(I_j)$, so direction change in the gap is in this range. Since $I_j\subseteq \{-v+1, \dots, v-1\}$ and $R_j\subseteq \{1, \dots, v\}$, the left side of the gap is at least $(2v-1)(j-1)+v-(v-1)-(v-1)=(2v-1)(j-1)-v+2$ and at most $(2v-1)(j-1)+v-(-v+1)-0=(2v-1)(j-1)+2v-1$.

By \Cref{lem:leftandrightsidelengths}, the left side of a gap in state $i$ has left length $(2v-1)(\sigma(i)-1)+x$ for some $x\in R_i\subseteq \{1, \dots, v\}$. The inequalities $(2v-1)(j-1)-v+2\le (2v-1)(\sigma(i)-1)+x \le (2v-1)(j-1)+2v-1$ can only be satisfied when $\sigma(i)=j$. So if a row is in state $j$ and the row above it is in state $i$, then $j=\sigma(i)$, proving the claim.
\end{proof}

Finally, we describe the relationship between the values in one row and the values in the row below it.

\begin{lemma}\label{lem:intervalbound}
If a gap in a schematic has state $i$ and left value $x$ and right value $y$, then the gap below it to the left has right value $x'$ and the gap below it to the right has left value $y'$, where $x+y=x'+y'$ and $x'-x=y-y'\in I_{\sigma(i)}$. 
\end{lemma}

\begin{proof}[Proof of \Cref{lem:intervalbound}]
By \Cref{lem:statecycle}, the rows in the schematic cycle through the states, and a gap with state $i$ has left value in $L_i$ and right value in $R_i$ by \Cref{lem:leftandrightsidelengths}.

Consider a gap $X$ in state $i$ with left value $x$ and right value $y$.
So by definition it has left length $(2v-1)(\sigma(i)-1)+x$ and right length $(2v-1)(s-\sigma(i))+y$, as in \Cref{fig:intervalbound}.
By \Cref{lem:statecycle}, $X$ is above bucket $\sigma(i)$ in some block. Let $Y$ be the gap below $X$ to the left and let $Z$ be the gap below $X$ and to the right.
As seen in the proof of \Cref{lem:leftandrightsidelengths}, $Y$ has right value $x'$ where $v-x'$ is the relative position of the left edge of $X$ in the bucket below it.
Similarly, $Z$ has left value $y'$ where $\ell_b+y'-v$ is the relative position of the right edge of $X$. So the length of $X$ is $\ell_b-2v+x'+y'$. It is also $(2v-1)(s-1)+x+y=\ell_b-2v+x+y$, so $x'+y'=x+y$.

The direction change in $X$ must occur between the arrow markings in bucket $\sigma(i)$, which face relative positions $(2v-1)(\sigma(i)-1)+v-\max(I_{\sigma(i)})$ and $(2v-1)(\sigma(i)-1)+v-\min(I_{\sigma(i)})$. The left side of $X$ starts at relative position $v-x'$ and has length $(2v-1)(\sigma(i)-1)+x$. So $x'-x\in I_{\sigma(i)}$. 
\end{proof}

\begin{figure}
    \centering
    \includegraphics[page=12]{TwoTiling.pdf}
    \caption{Calculations for \Cref{lem:intervalbound}}
    \label{fig:intervalbound}
\end{figure}

The \emph{gap conditions} are the constraints described by \Cref{lem:leftandrightsidelengths,lem:statecycle,lem:intervalbound}. Specifically, these are:

\begin{itemize}
    \item If a row has state $i$, then the row below it has state $\sigma(i)$
    \item If a row has state $i$, then the gaps in that row have state $i$
    \item The left value of a gap in state $i$ is in $L_i$ and the right value of a gap in state $i$ is in $R_i$
    \item If a gap is in state $i$ and has left value $x$ and right value $y$, the gap below it to the left has right value $x'$, and the gap below it to the right has left value $y'$, then $x, y, x', y'$ satisfy $x+y=x'+y'$ and $x'-x=y-y'\in I_{\sigma(i)}$.
\end{itemize}

\begin{lemma}\label{lem:gapconditions}
Any valid schematic has gaps satisfying the gap conditions.

Any assignment of states and left and right values of gaps that satisfies the gap conditions corresponds to a valid schematic.
\end{lemma}

\begin{proof}
It is clear by \Cref{lem:leftandrightsidelengths,lem:statecycle,lem:intervalbound} that any valid schematic satisfies the gap conditions. 

The left and right lengths of a gap are determined by its state and its left and right values. So given this data for each gap, we can construct each row. For a row $X$ in state $i$ and a row $Y$ in state $\sigma(i)$, the row $X$ fits above row $Y$ if the gaps in both rows satisfy the gap conditions (this by essentially the same calculations as in \Cref{lem:leftandrightsidelengths,lem:statecycle,lem:intervalbound}).
\end{proof}

\subsection{Structure of states representing an AB tiling}

We will choose an even number of buckets, so if $i$ is even (resp. odd), then $\sigma(i)$ is odd (resp. even). The gaps form a diagonal lattice. We label each gap with a row and column, as shown in \Cref{fig:rowsandcolumns}. The index of a row increases by $1$ when moving down and the index of column increases by $1$ when moving right. We write $(r, c)$ for the gap in row $r$ and column $c$. As a shorthand, we sometimes write $(r, c)_L$ for the left value of that gap and $(r, c)_R$ for the right value. We write that a gap has values $(x|y)$ if it has left value $x$ and right value $y$. 

\begin{figure}
    \centering
    \includegraphics[page=13]{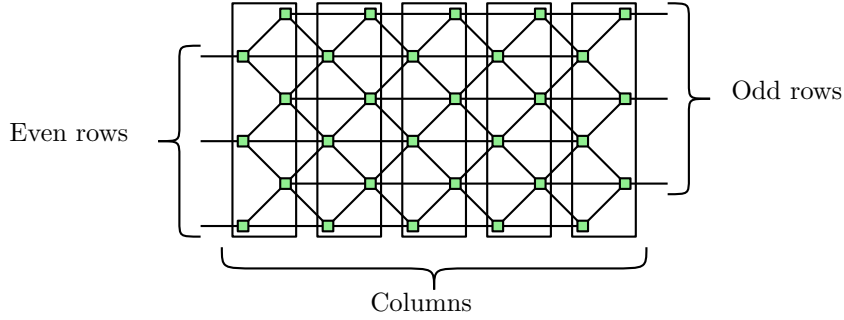}
    \caption{If a gap is in an odd row, then the one below it to the left is in the same column. If a gap is in an even row, then the one below it to the right is in the same column.}
    \label{fig:rowsandcolumns}
\end{figure}

A complete schematic will correspond to an AB-tiling with roughly the same geometric structure as the schematic. Each column in the schematic will represent a column of AB tiles in an AB-tiling. Moving left or right by one gap in the schematic corresponds to moving in the same direction by $1$ unit in an AB-tiling. It takes several states worth of operations to check the constraints of an AB-tiling, so there are multiple rows in the schematic corresponding to each row in the AB tiling. Every time the rows complete a full cycle through the states, this corresponds to moving down by a distance of $2$ units in the AB-tiling.

Now fix a set of AB tiles with tile sets $A$ and $B$. Let $t=|A|=|B|$ and set $v=(2t+1)d$, where $d\ge 2$ is a constant that will be set later.

We will designate some important states called \emph{encoding states}. An encoding state $e$ is an odd state that has $L_e=\{di : 1\le i\le 2t\}$ and $R_e=\{d(2t+1-i) : 1\le i\le 2t\}$. In an encoding state, a left value of $di$ represents the $i$th $A$ tile if $1\le i\le t$ or the $(i-t)$th $B$ tile if $t+1\le i\le 2t$. 

When we construct a schematic corresponding to an AB tiling, an encoding state with left value $di$ will have right value $d(2t+1-i)$. But there might exist valid schematics where this is not the case. So when constructing an AB tiling given a schematic, we should only use the left value to determine the tile represented, and we should not assume that the right value encodes useful information.

Each interval between encoding states forms a gadget that creates some part of the AB tiling structure. There are $4$ main encoding states which correspond to parts of an AB tiling, as shown in \Cref{fig:rowcorrespondence}. We write $e_1$ through $e_4$ for the indices of these states.

\begin{figure}
    \centering
    \includegraphics[page=14]{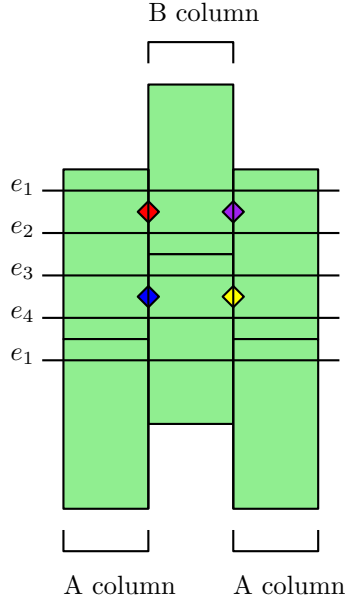}
    \caption{The encoding states $e_1$ through $e_4$ correspond to different parts of an AB tiling. States $e_1$ and $e_2$ go through the upper parts of the $A$ tiles and the lower parts of the $B$ tiles, while states $e_3$ and $e_4$ go through the lower parts of the $A$ tiles and the upper part of the $B$ tiles.}
    \label{fig:rowcorrespondence}
\end{figure}

We will set $e_1=1$ and $e_2=9$. These values in these states represent the top square of an $A$ tile or the bottom square of a $B$ tile. The states between $e_1$ and $e_2$ create a constraint that causes columns to alternate between $A$ and $B$, and requires top colors of an $A$ tile to match the corresponding bottoms colors on adjacent $B$ tiles. The tile represented by a column should not change between $e_1$ and $e_2$. This constraint will be described in \Cref{sec:horizontalconstraints}.

The next encoding states are $e_3=9+4(t-1)$ and $e_4=17+4(t-1)$. The states between $e_3$ and $e_4$ are similar to the ones between $e_1$ and $e_2$, except that now they require the bottom colors of an $A$ tile to match the corresponding top colors on the adjacent $B$ tiles. 

The states between $e_2$ and $e_3$ keep the index of an $A$ tile fixed, but allow the index of a $B$ tile to change. If a column represents an $A$ tile in states $e_1$ and $e_2$, then it should represent the same $A$ tile in states $e_3$ and $e_4$. However, a column that represents a $B$ tile in states $e_1$ and $e_2$ should be allowed to represent any $B$ tile in states $e_3$ and $e_4$. The states between $e_2$ and $e_3$ will be described in \Cref{sec:verticalconstraints}.

There are an additional $4(t-1)-1$ states after $e_4$ before looping back to $e_1$ (so there are $16+8(t-1)$ states total). These are similar to those between $e_2$ and $e_3$, but this time they preserve the index of tiles in a $B$ column but allow the index of a tile in an $A$ column to change. 

\subsection{Creating the AB pattern}\label{sec:verticalconstraints}

Here, we describe the states between $e_2$ and $e_3$. The states between $e_4$ and $e_1$ can be created in an analogous way.

Recall that $e_2$ is an encoding state, so has $L_{e_2}=\{dj:1\le j\le 2t\}$. If $(r+e_2, c)$ is in state $e_2$ and represents an $A$ tile, then $(r+e_3, c)$ should represent the same $A$ tile.
But if $(r+e_2, c)$ represents a $B$ tile, then $(r+e_3, c)$ should be allowed to represent any $B$ tile.
\Cref{lem:verticalonlyconstraint} provides a basic building block for creating these constraints.

\begin{lemma}\label{lem:verticalonlyconstraint}
Suppose that $i$ is odd, $R_{\sigma(i)}=L_{\sigma(i)}=\{1, \dots, v\}$, and $I_{\sigma^2(i)}=\{0\}$. Let $r$ be a row in state $i$ and say that gap $(r, c)$ has values $(x_c|y_c)$ for each $c\in \mathbb{Z}$.

Then for each $c\in \mathbb{Z}$, the values $(x_c'|y_c')$ of $(r+2, c)$ satisfy $x_c'+y_c'=x+y$ and $x_c'-x_c\in I_{\sigma(i)}$. 

Furthermore, for any assignment of $x_{c}'\in L_{\sigma^2(i)}$ and $y_{c}'\in R_{\sigma^2(i)}$ satisfying the above, we can fill rows $r+1$ and $r+2$ so that gap $(r+2, c)$ has values $(x_c'|y_c')$ in a way that satisfies the gap conditions.
\end{lemma}

\begin{proof}
Let $(y_{c-1}''|x_c'')$ be the values of gap $(r+1, c)$.
By the gap conditions, we have $x_c''+y_c''=x_c+y_c$ and $x_c''-x_c\in I_{\sigma(i)}$; see  \Cref{fig:verticalonlyconstraint}.
Since $I_{\sigma^2(i)}=\{0\}$, the values $(x_c'|y_c')$ of gap $(r+2, c)$ must be $x_c'=x_c''$ and $y_c'=y_c''$. So $x_c'+y_c'=x_c+y_c$ and $x_c'-x_c\in I_{\sigma(i)}$.

Given such values of $(x_c'|y_c')$ of $(r+2, c)$, we can set the values of $(r+1, c)$ to be $(y_{c-1}'|x_c')$. This satisfies the gap conditions.

\begin{figure}
    \centering
    \includegraphics[page=15]{TwoTiling.pdf}
    \caption{Illustration of \Cref{lem:verticalonlyconstraint}.}
    \label{fig:verticalonlyconstraint}
\end{figure}
\end{proof}

Each odd $i$ with $e_2\le i<e_3$ will be as in the conditions of \Cref{lem:verticalonlyconstraint}. So we just need to specify $L_i$ and $R_i$ for odd $i$ and $I_i$ for even $i$.

The idea can be represented by a diagram as in \Cref{fig:verticalconstraints1}. Each tick mark in the figure represents one of the values of $L_i$ for odd $e_2\le i\le e_3$. In order to fill in a single column, we need to choose one left value for each such state, corresponding to a choice of one tick mark in the figure for each of the rows shown. These choices must satisfy the conditions of \Cref{lem:verticalonlyconstraint}, which is illustrated by the intervals shown in the figure. \Cref{fig:verticalconstraints1.1} shows a valid assignment, where each tick mark chosen is below the interval for the chosen tick mark in the previous row.

If one of the blue tick marks (corresponding to an $A$ tile) is chosen in the first row, then there is only one path through the figure, meaning the the index of the tile represented is kept fixed. If one of the red tick marks is chosen instead (corresponding to a $B$ tile), then there are several valid paths through the figure, allowing any of the red tick marks in the bottom row to be chosen. So the index of a $B$ tile can can change.

\begin{figure}
    \centering
    \includegraphics[page=16]{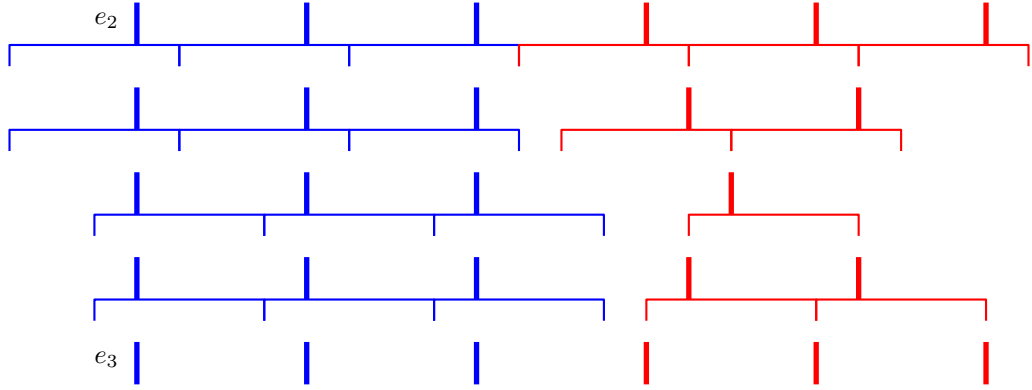}
    \caption{
    Each vertical tick represents one of the values in $L_i$ for odd $i$ going from $e_2$ (top row) to $e_3$ (bottom row). The blue ticks correspond to $A$ tiles and the red ticks correspond to $B$ tiles. Below each tick, there is copy of the interval $I_{i+1}$ which is translated so that $0$ is aligned with that tick. By \Cref{lem:verticalonlyconstraint}, the difference between the left value in $(r, c)$ and the one in $(r+2, c)$ is in $I_{i+1}$. So the translated copy of this interval is above the left values that $(r+2, c)$ can take when $(r, c)$ has the left value corresponding to that interval. This figure shows the case where $t=3$. 
    }
    \label{fig:verticalconstraints1}
\end{figure}

\begin{figure}
    \centering
    \includegraphics[page=17]{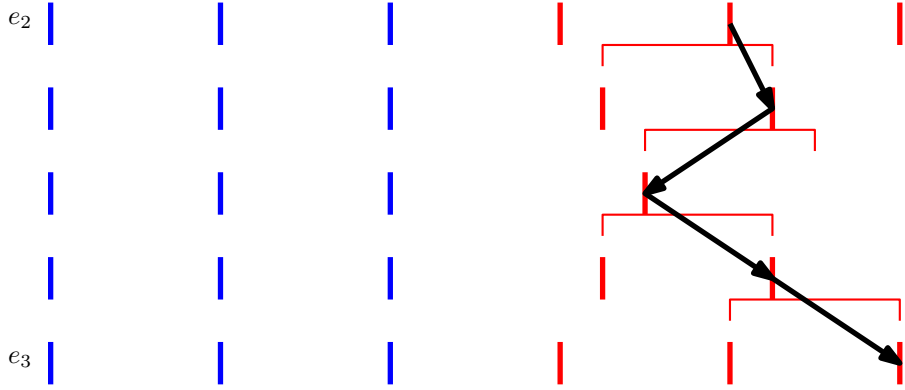}
    \caption{In order to fill in a single column in the states between $e_2$ and $e_3$, we choose a path through \Cref{fig:verticalconstraints1} starting at one of the left values in $e_2$ and ending at one of the left values in $e_3$. The intervals shown in the figure determine which edges are allowed in this path. An example of such a path is shown.}
    \label{fig:verticalconstraints1.1}
\end{figure}

Recall that $e_2=9$ and $e_3=9+4(t-1)$. The states between $e_2$ and $e_3$ satisfy the assumptions of \Cref{lem:verticalonlyconstraint}. Specifically, each even state $e_2\le \alpha\le e_3$ has $R_\alpha=L_\alpha=\{1, \dots, v\}$ and each odd state $e_2<\alpha\le e_3$ is given $I_\alpha=\{0\}$. Additionally, we set $R_{\alpha}=\{1, \dots, v\}$ for each odd $e_2\le \alpha\le e_3$. The remaining information about these states is given by the following:

\begin{center}
\begin{tabular}{|c|c|}
\hline
Even state $\alpha$&$I_\alpha$\\
\hline
$e_2+2i-1$ for $1\le i \le t-1$&$[-d+1, 1]$\\
$e_2+2(t-1)+2i-1$ for $1\le i \le t-1$&$[-1, d-1]$\\
\hline
\end{tabular}
\end{center}

\begin{center}
\begin{tabular}{|c|c|c|}
\hline
Odd state $\alpha$&$L_\alpha$\\
\hline
$e_2+2i$ for $1\le i \le t-1$&$\{d\ell : 1\le \ell \le t\}\cup\{d\ell+i : t+1\le \ell\le 2t-i\}$\\
$e_2+2(t-1)+2i$ for $1\le i \le t-1$&$\{d\ell: 1\le \ell \le t\}\cup \{d\ell+t-1-i: t+1\le \ell \le t+1+i\}$\\
\hline
\end{tabular}
\end{center}

\Cref{lem:verticalspecification} tells us how to construct the gaps between $e_2$ and $e_3$ when building a tiling. 

\begin{lemma}\label{lem:verticalspecification}
Suppose the states between $e_2$ and $e_3$ as above. Let $r$ be a row such that $r+e_2$ is in state $e_2$.
For each $j$, suppose that gap $(r+e_2, j)$ has values $(dv_j|d(2t+1-v_j))$ for some $1\le v_j \le 2t$. Then the rows up to $r+e_3$ can be specified so that the gap $(r+e_3, j)$ has values $(dw_j|d(2t+1-w_j))$ for any assignment of $1\le w_j \le 2t$ satisfying:

\begin{itemize}
    \item If $1\le v_j\le t$, then $w_j=v_j$
    \item If $t+1\le v_j \le 2t$, then $t+1\le w_j\le 2t$
\end{itemize}
\end{lemma}

\begin{proof}
We just need to specify the values for rows $r+e_2+2i$ for $1\le i \le 4(t-1)$ that satisfy the conditions of \Cref{lem:verticalonlyconstraint}. 

For values of $j$ with $1\le v_j \le t$, we have $w_j=v_j$. So for such a $j$ and for $1\le i\le 2(t-i)$, gap $(r+e_2+2i|j)$ is given values $(dv_j|d(2t+1-v_j)$.

For values of $j$ with $t+1\le v_j\le 2t$, we have $t+1\le w_j\le 2t$. For such $j$ and for $1\le i \le 2(t-1)$, gap $(r+e_2+2i, j)$ is given left value:

\[
\begin{cases}
d(v_j-i)+i&\text{for }i\le v_j-t-1\\
d(t+1)+i&\text{for }i\ge v_j-t-1\\
\end{cases}
\]

\noindent while gap $(r+e_2+4(t-1)+2i, j)$ is given left value:

\[
\begin{cases}
d(t+i)+t-1-i&\text{for }i\le w_j-t\\
dw_j+t-1-i&\text{for }i\ge w_j-t\\
\end{cases}
\]

\noindent The right value of these gap are set so that the left and right values of each gap sum to $d(2t+1)$. 

It is straightforward to check that this satisfies the conditions of \Cref{lem:verticalonlyconstraint}.
\end{proof}

\Cref{lem:verticalverification} verifies that that the states between $e_2$ and $e_3$ enforce the appropriate constraints.

\begin{lemma}\label{lem:verticalverification}
Suppose the states between $e_2$ and $e_3$ are as above. Suppose that the rows $r+e_2$ through $r+e_3$ are filled, where the left value in gap $(r+e_2, j)$ is $v_j$. If $1\le v_j\le t$, then the left value in gap $(r+e_3, j)$ is $dv_j$. If $t+1\le v_j\le 2t$, then the left value in gap $(r+e_3, j)$ is $dw_j$ for some $t+1\le w_j\le 2t$.
\end{lemma}

\begin{proof}
First, consider some $j$ where $1\le v_j\le t$. In each of the odd states between $e_2$ and $e_3$, the only value $x\in L_i$ satisfying $dv_j-d+1\le x\le dv_j+d-1$ is $x=dv_j$. The interval $I_i$ is contained in $\{-d+1, d-1\}$ for every even state between $e_2$ and $e_3$. By \Cref{lem:verticalonlyconstraint} and induction, we see that the left value in gap $(r+i, j)$ is $v_j$ for each odd $i$ with $e_2<i\le e_3$. In particular, the left value in gap $(r+e_3, j)$ is $v_j$.

Now consider some $j$ where $t+1\le v_j\le 2t$. In each of the odd states between $e_2$ and $e_3$, the smallest $x\in L_i$ with $x\ge d(t+1)-d+1$ is at least $d(t+1)$. By induction and \Cref{lem:verticalonlyconstraint}, we see that the left value in gap $(r+i, j)$ is at least $d(t+1)$ for each $i$ with $e_2<i\le e_3$. Since $L_{e_3}=\{d\ell : 1\le \ell\le 2t\}$, we see that the left value in gap $(r+e_3, j)$ is $dw_j$ for some $t+1\le w_j\le 2t$.
\end{proof}

The states between $e_4$ and $e_1$ can be created in a way analogous to those between $e_2$ and $e_3$. The difference is that the roles of the $A$ and $B$ tiles are swapped, so the index of an $A$ tile is allowed to change while the index of a $B$ tile is fixed. A diagram of this is shown in \Cref{fig:verticalconstraints2}.

\begin{figure}
    \centering
    \includegraphics[page=18,width=\textwidth]{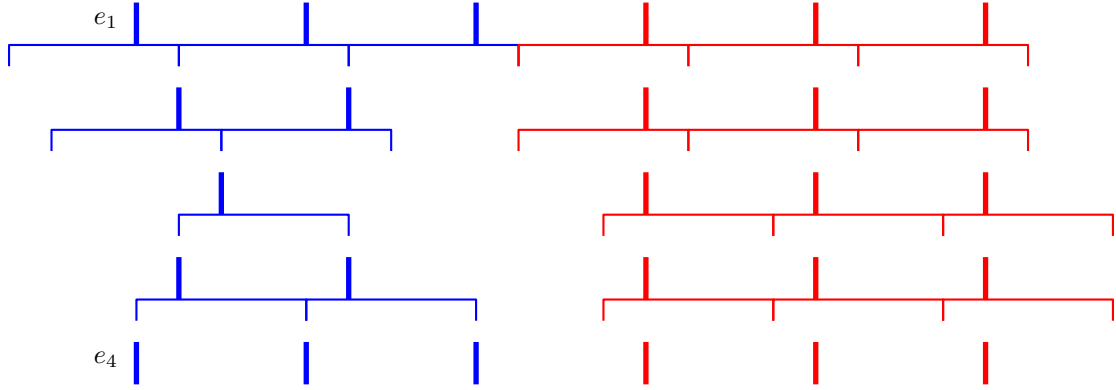}
    \caption{The analogous diagram to \Cref{fig:verticalconstraints1}, but for the states between $e_4$ and $e_1$.}
    \label{fig:verticalconstraints2}
\end{figure}

\begin{lemma}\label{lem:verticalconstraints}
We can construct states between $e_2$ and $e_3$ such that, if a row $r+e_2$ has state $e_2$ and row $r+e_3$ has state $e_3$, then the rows between $r+e_2$ and $r+e_3$ can be filled in if and only if:

\begin{itemize}
    \item If a column represents an $A$ tile in row $r+e_2$, then it represents the same $A$ tile in row $r+e_3$
    \item If a column represents a $B$ tile in row $r+e_2$, then it represents some $B$ tile in row $r+e_3$
\end{itemize}

We can construct the states between $e_4$ and $e_1$ in an analogous way, but with the roles of $A$ and $B$ tiles reversed.
\end{lemma}

\begin{proof}
Clear by \Cref{lem:verticalspecification,lem:verticalverification}.
\end{proof}

\subsection{Creating color-matching constraints}\label{sec:horizontalconstraints}

We want to create constraints that require colors on opposite sides of a vertical edge to match. Let $k$ be the number of colors in the AB tiling instance, and label the colors $0$ through $k-1$. We now fix\footnote{Our Lean formalization uses $d=12(k+t)$ instead. This does not materially affect the proof.} $d=12k$. 

We split the space between $e_1$ and $e_2$ by adding another encoding state $e_{1.5}=5$ between them. The states between $e_1$ and $e_{1.5}$ will require the upper-right color of an $A$ column to match the lower-left color of the $B$ column right of it, and the states between $e_{1.5}$ and $e_2$ will require the upper-left color of an $A$ column to match the lower-right color of the $B$ column left of it. An encoding state $e_{3.5}$ is added between $e_3$ and $e_4$ analogously. This is shown in \Cref{fig:detailedrowcorrespondence}.

\begin{figure}
    \centering
    \includegraphics[page=19]{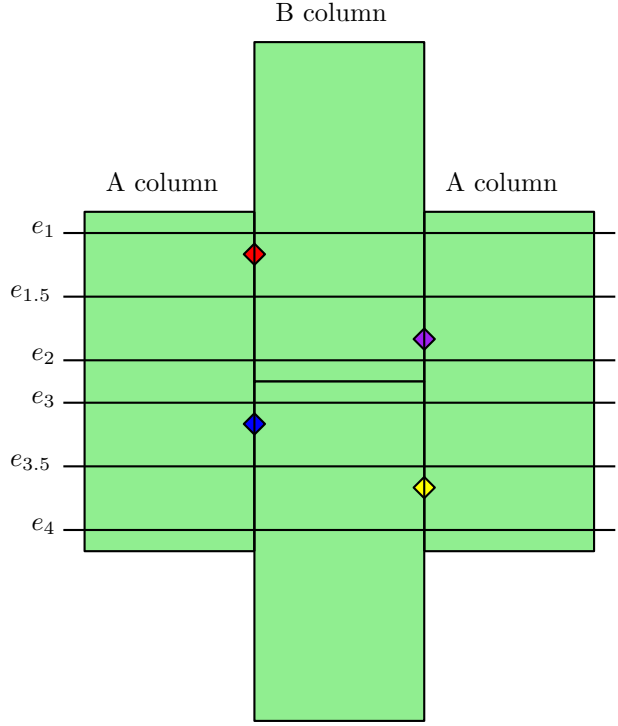}
    \caption{The gaps between $e_1$ and $e_2$ and between $e_3$ and $e_4$ are each split in half.}
    \label{fig:detailedrowcorrespondence}
\end{figure}

In this section, we describe the states between $e_1$ and $e_{1.5}$. The other states can be created analogously. 

States $e_1$ and $e_{1.5}$ are encoding states, so have $L_*=\{12kj:1\le j\le 2t\}$ and $R_*=\{12k(2t+1-j):1\le j \le 2t\}$. We set $I_{e_1+1}=I_{e_1+2}=I_{e_1+3}=I_{e_{1.5}}=[0, 6k]$. For $1\le j\le 2t$, let $c(j)$ be the upper-right color when $j$ is an $A$ tile or the lower-left color when $j$ is a $B$ tile. The values in $L_*$ and $R_*$ for $e_1+1$ through $e_1+3$ are then given by:

\begin{center}
\begin{tabular}{|c|c|c|}
\hline
State $i$&$L_i$&$R_i$\\
\hline
$e_1+1$&$\{1, \dots, v\}$&
$\begin{matrix}
\{12kj+2k+c(j) : 1\le j\le t\}\cup \\
\{12kj+4k+c(j) : t+1\le j \le 2t\}
\end{matrix}$
\\
\hline
$e_1+2$&
$\{12ki : 1\le i\le 2t\}$&
$\begin{matrix}
\{12k(2t+1-j)+2k : 1\le j \le 2t\}\cup \\ 
\{12k(2t+1-j)-2k+c: 1\le j \le 2t \text{ and } -k+1\le c\le k-1\}
\end{matrix}$
\\
\hline
$e_1+3$&$\{1, \dots, v\}$&$\{12ki+x : 0\le x < 6k\}$\\
\hline
\end{tabular}
\end{center}

\Cref{lem:horizontalspecification} shows that the gaps between $e_1$ and $e_{1.5}$ can be filled when constructing a tiling.

\begin{lemma}\label{lem:horizontalspecification}
Suppose that row $r$ is in state $e_1$ and suppose also that for each $j$, gap $(r, j)$ has values $(12kv_j|12k(2t+1-v_j))$. If each even $j$ has $1\le v_{j}\le t$, $t+1\le v_{j+1}\le 2t$ and $c(v_{j})=c(v_{j+1})$, then the rows up to $r+e_{1.5}$ can be filled in such a way that gap $(e_{1.5}|j)$ has values $(12kv_j|12k(2t+1-v_j))$.
\end{lemma}

\begin{proof}
Consider some even $j$, so $1\le v_j\le t$ and $t+1\le v_{j-1},v_{j+1}\le 2t$ by assumption. \Cref{fig:horizontalspecification1} shows how columns $j$ and $j+1$ are filled, while \Cref{fig:horizontalspecification2} shows how columns $j-1$ and $j$ are filled. It is straightforward to check that these assignments satisfy the gap conditions.

\begin{figure}
    \centering
    \includegraphics[page=21]{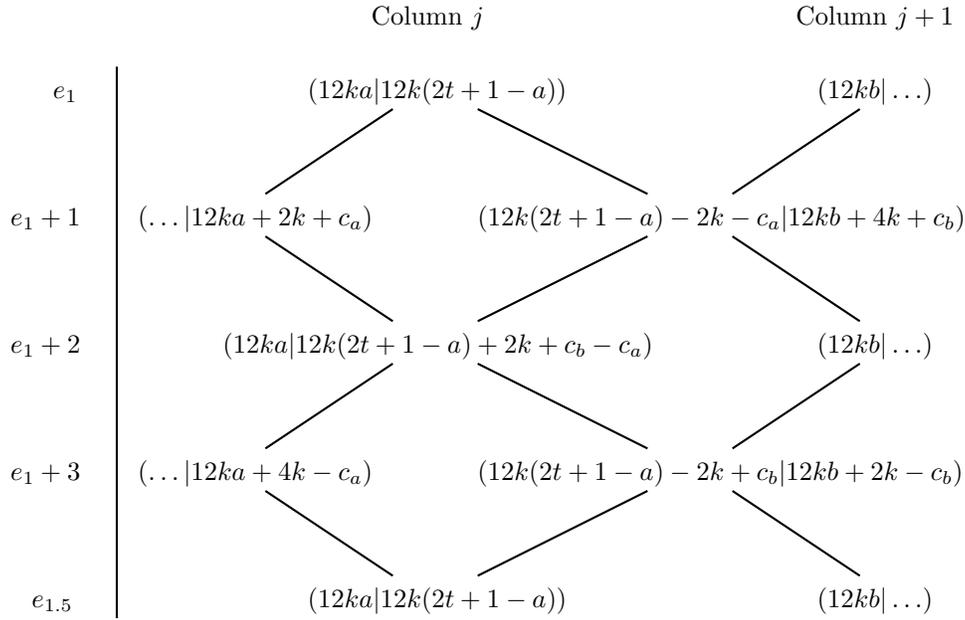}
    \caption{Filling in the columns $j$ and $j+1$ when $v_j$ represents an $A$ tile. In the figure, $a=v_j$ and $b=v_{j+1}$. Note that $c_a=c_b$ by assumption, so $12k(2t+1-a)+2k+c_b-c_a$ is in $R_{e_1+2}$.}
    \label{fig:horizontalspecification1}
\end{figure}

\begin{figure}
    \centering
    \includegraphics[page=20]{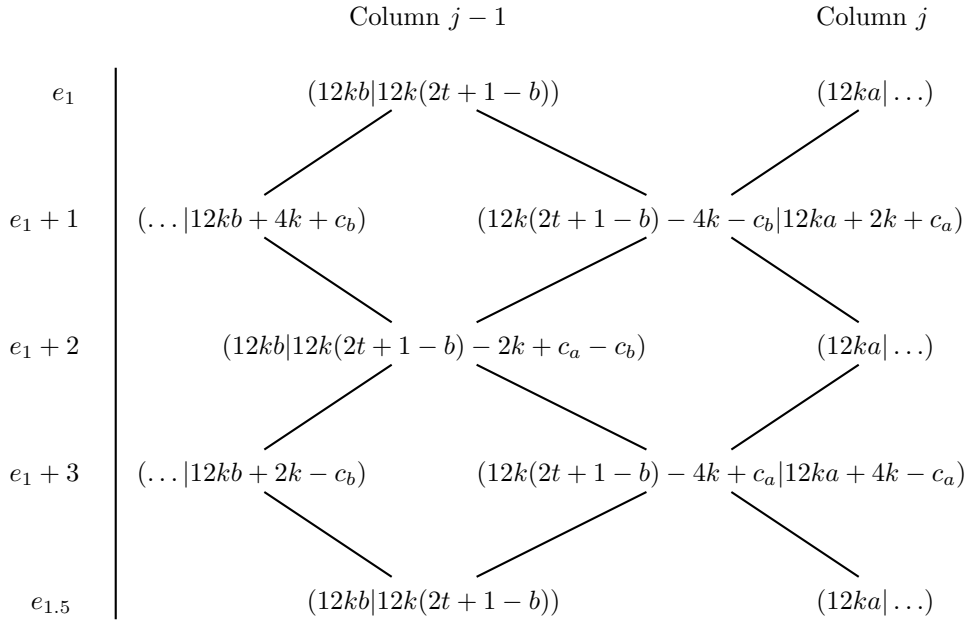}
    \caption{Filling in the columns $j-1$ and $j$ when $v_j$ represents and $A$ tile. In the figure, $a=v_j$ and $b=v_{j+1}$. The colors $c_a$ and $c_b$ might not match, which is okay because $12k(2t+1-b)-2k+c_a-c_b$ is in $R_{e_1+2}$.}
    \label{fig:horizontalspecification2}
\end{figure}
\end{proof}

\Cref{lem:horizontalverification} verifies that the states $e_1$ to $e_{1.5}$ enforce the appropriate constraint.

\begin{lemma}\label{lem:horizontalverification}
Suppose there is a complete schematic and that row $r+e_1$ is in state $e_1$, where for each $j$ the gap $(r+e_1, j)$ has left value $12kv_j$ for some $1\le v_j\le 2t$. Then:

\begin{itemize}
    \item When $1\le v_j \le t$, we have $t+1\le v_{j+1}\le 2t$ and $c(v_j)=c(v_{j+1})$
    \item When $t+1\le v_j\le 2t$, we have $1\le v_{j+1}\le t$
\end{itemize}

Furthermore, gap $(r+e_{1.5}, j)$ has left value $12kv_j$. 
\end{lemma}

\begin{proof}
Gap $(r+e_1, j)$ has left value $12kv_j$ by assumption. By the gap conditions, $(r+e_1+1, j)_R\in R_{e_1+1}$ and $(r+e_1+1, j)_R-(r+e_1, j)_L\in I_{e_1+1}=[0, 6k]$. So $12kv_j\le (r+e_1+1, j)_R\le12kv_j+6k$. For a given $j$, there is only one value in $R_{e_1+1}$ in this range, which is:

\begin{equation}\label{eqn:colorloading}
(r+e_1+1, j)_R=\begin{cases}
    12kv_j+2k+c(v_j)&\text{for } 1\le v_j\le t\\
    12kv_j+4k+c(v_j)&\text{for } t+1\le v_j\le 2t\\
\end{cases}
\end{equation}

Similarly, $(r+e_1+2, j)_L\in L_{e_1+2}$ and $(r+e_1+1)_R-(r+e_1+2, j)_L\in I_{e_1+2}=[0, 6k]$. So $12kv_j+qk+c(v_j)-6k\le (r+e_1+2, j)_L\le 12kv_j+qk+c(v_j)$. This gives:

\begin{equation}\label{eqn:colorreset}
(r+e_1+2, j)_L=12kv_j
\end{equation}

By similar calculations, we obtain $(r+e_1+3, j)_R=12kv_j+x$ for some $0\le x\le 6k$ and $(r+e_1+4, j)_L=(r+e_{1.5}, j)_L=12kv_j$. 

Since $e_1$ is a encoding state, the right value of gap $(r+e_1, j)$ is $12k(2t+1-r_j)$ for some $1\le r_j\le 2t$. It may or may not be the case that $r_j=v_j$. 

Consider some $j$ with $t+1\le v_j\le 2t$, and suppose for contradiction that $t+1\le v_{j+1}\le 2t$. By the gap conditions, $(r+e_1+1, j)_R+(r+e_1+1, j+1)_L=(r+e_1, j)_L+(r+e_1, j)_R=12k(2t+1+v_j-r_j)$. By \eqref{eqn:colorloading}, $(r+e_1+1, j)_R=12kv_j+4k+c(v_j)$, so $(r+e_1+1, j)_L=12k(2t+1-r_j)-4k-c(v_j)$. 

By the gap conditions, $(r+e_1+1, j+1)_L+(r+e_1+1, j+1)_R=(r+e_1+2, j)_R+(r+e_1+2, j+1)_L$. By \eqref{eqn:colorloading}, $(r+e_1+1, j+1)_R=12kv_{j+1}+4k+c(v_{j+1})$. By \eqref{eqn:colorreset}, $(r+e_1+2, j+1)_L=12kv_{j+1}$. So $(r+e_1+2, j)_R=12k(2t+1-r_j)+c(v_{j+1})-c(v_j)$. Recall that $c(v_j)$ and $c(v_{j+1})$ are between $1$ and $k$, so $12k(2t+1-r_j)+c(v_{j+1})-c(v_j)$ is not in $R_{e_1+2}$, violating the gap conditions. By contradiction we conclude that $1\le v_{j+1}\le t$.

By an analogous argument, we can show that if $1\le v_j\le t$, then $t+1\le v_{j+1}\le 2t$. In this case, we calculate that $(r+e_1+2, j)_R$ is $12k(2t+1-r_j)+2k+c(v_{j+1})-c(v_j)$. This must be in $R_{e_1+2}$, which is only possible if $c(v_j)=c(v_{j+1})$.
\end{proof}

The states between $e_{1.5}$ and $e_{2}$ are created in an analogous way, only the roles of $A$ and $B$ tiles are reversed and the colors used are the upper-left color of an $A$ tile or the lower-right color of a $B$ tile. The states between $e_2$ and $e_3$ are created similarly to those between $e_{1.5}$ and $e_2$, but using the lower colors of $A$ tiles and the upper colors of $B$ tiles.

\subsection{Final verification}

\begin{proof}[Proof of \Cref{thm:onetiling}]
Let $\mathcal{S}=\{A, B\}$ be an AB-tiling instance. Then we can construct a single prototile $T$ and a set of edge-to-edge matching rules as we have described .

First, suppose that $\mathcal{S}$ admits an AB-tiling, and fix such a tiling. Let $t=|A|=|B|$. For $r, c\in \mathbb{Z}$, let $a_{rc}\in \{1, \dots, t\}$ be the index of the $A$ tile with lower-left corner $(2c, 2r)$ and let $b_{rc}$ be the index of the $B$ tile with lower-left corner $(2c+1, 2r+1)$.

We start by specifying the states and values of the gaps in a schematic. For $i\in \{1, 2, 3, 4\}$, the gap $(e_i-sr, 2c)$ is given values $(da_{rc}|d(2t+1-a_{rc}))$. Gaps $(e_1-sr, 2c+1)$ and $(e_2-sr, 2c+1)$ are given values $(db_{rc}|d(2t+1-b_{rc}))$ while gaps $(e_3-sr, 2c+1)$ and $(e_4-sr, 2c+1)$ are given values $(db_{r-1, c}|d(2t+1-b_{r-1,c})$. By \Cref{lem:verticalconstraints,lem:horizontalspecification}, the rest of the gaps can be filled in a way that satisfies the gap conditions. This assignment can be used to generate a valid schematic, which can be used to construct a weave-pattern tiling with copies of $T$. 

Now, suppose that $T$ tiles the plane. By \Cref{lem:weavepattern}, this tiling forms a weave pattern. As discussed in \Cref{sec:schematics}, this can be interpreted as a schematic. As shown in \Cref{sec:gapconditions}, the gaps in this schematic satisfy the gap conditions. Choose a labeling of the rows and columns so that row $e_1$ is in state $e_1$ and so that $(e_1, 0)$ represents and $A$ tile (this is possible by \Cref{lem:statecycle,lem:horizontalverification}). By \Cref{lem:horizontalverification,lem:verticalconstraints}, for $r, c\in \mathbb{Z}$, $(e_1-sr, 2c)$ always represents an $A$ tile and $(e_1-sr, 2c+1)$ always represents a $B$ tile. So we construct an AB-tiling where the $A$ tile with lower-left corner $(2c, 2r)$ is the one represented by gap $(e_1-sr, 2c)$ and the $B$ tile with lower-left corner is the one represented by gap $(e_1-sr, 2c+1)$. By \Cref{lem:horizontalverification}, the colors match along vertical edges, so this forms a valid AB tiling.

So $T$ tiles the plane if and only if $\mathcal{S}$ does. By \Cref{lem:ABtiling}, this proves the claim.
\end{proof}

\section{Applications of \Cref{thm:onetiling}}\label{sec:variants}

We have shown that no algorithm can decide if a stick tile can tile the plane when edge-to-edge matching rules are allowed. In this section, we will show that the addition of a single extra tile allows matching rules to be simulated geometrically, proving \Cref{thm:main}. We also sketch proofs that:

\begin{itemize}
    \item It is undecidable to determine if two prototiles tile the plane, even if those prototiles are polyhexes
    \item It is undecidable to determine if three polyomino prototiles tile the plane
    \item It is undecidable to determine if four polynomino prototiles tile the plane when only translational motion is allowed
\end{itemize}

All of these results work by simulating the construction from \Cref{sec:monotiling} using different sets of tiles.

\subsection{Simulating matching rules geometrically}\label{sec:stapletile}

First, we show how to simulate the matching rules by modifying the edges of the sticks and adding a second second \emph{staple} tile. Our staple tile is a regular $12$-gon, as in \Cref{fig:stapletile}. If the stick has length $n$, then the staple should have diameter smaller than $\frac{1}{8n+4}$-times the length of an edge of the stick. Any shape would work as a staple, so long as it is sufficiently small and cannot tile the plane independently. 

\begin{figure}
    \centering
    \includegraphics[page=1]{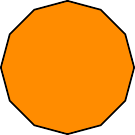}
    \caption{A regular $12$-gon.}
    \label{fig:stapletile}
\end{figure}

\begin{lemma}\label{lem:nostapletiling}
A regular 12-gon does not tile the plane on its own.
\end{lemma}

We modify the edges of the stick as follows. Label the edges $1$ through $4n+2$. For each edge, we designate $8n+4$ spots which have either a bump or a dent, as in \Cref{fig:dentbump}. Edge number $i$ has dents in the first $4n+2$ spots except for spot $i$, which has a bump. If a rule prevents edge $i$ from being placed against edge $j$, then $i$ has a bump in spot $8n+4-j+1$, otherwise it has a dent in that spot. An example is shown in \Cref{fig:matchingexample1}.

\begin{figure}
    \centering
    \includegraphics[page=2,width=\textwidth]{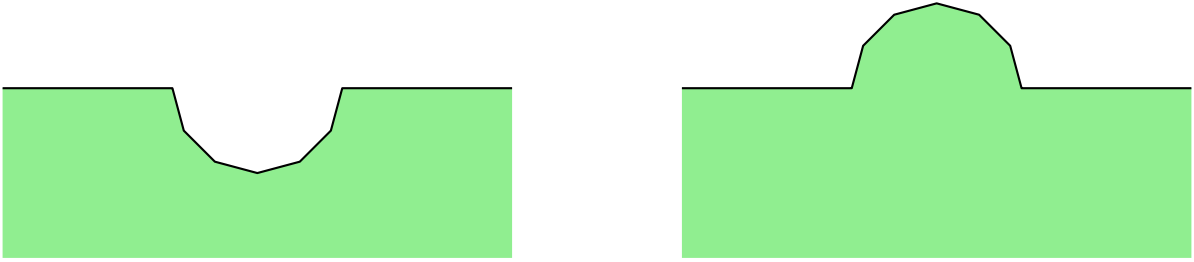}
    \caption{A dent (left) or a bump (right).}
    \label{fig:dentbump}
\end{figure}

\begin{figure}
    \centering
    \includegraphics[page=3]{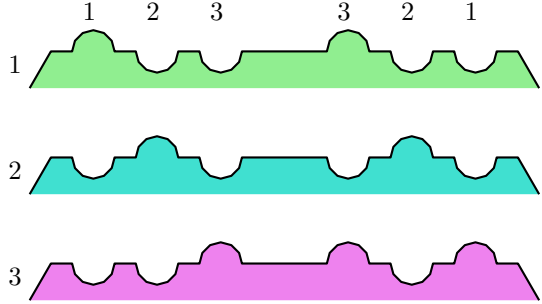}
    \caption{Three edges with dents and bumps. Here the prohibited pairs are $(1, 3)$, $(2, 2)$, and $(3, 3)$.}
    \label{fig:matchingexample1}
\end{figure}

When two edges that satisfy the matching rules are placed against each other, the bumps and dents line up, leaving spaces that can be filled by the staple, as in \Cref{fig:matchingexample2}. If a bump is placed against a dent, then the bump fills the dent. If two dents are placed against each other, then they create a space that can be exactly filled by a staple. However, if two edges are placed against each other in such a way that they would violate the matching rules, then two bumps would overlap, as in \Cref{fig:matchingexample3}.

\begin{figure}
    \centering
    \includegraphics[page=4]{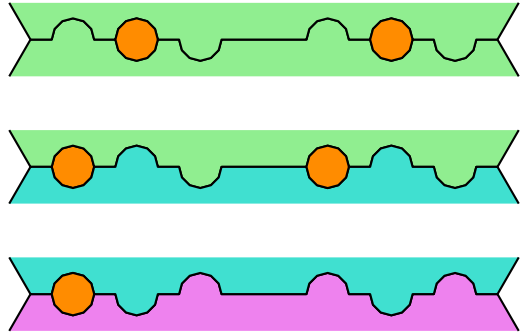}
    \caption{The allowed pairs $(1, 1)$, $(1, 2)$, and $(2, 3)$ from the example in \Cref{fig:matchingexample1}.}
    \label{fig:matchingexample2}
\end{figure}

\begin{figure}
    \centering
    \includegraphics[page=5]{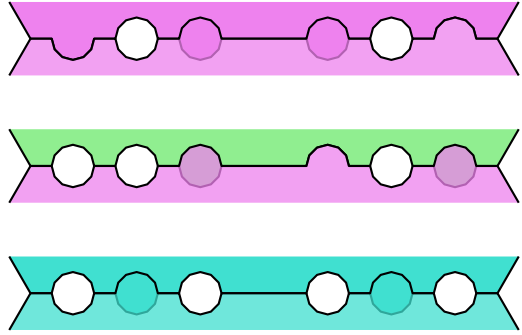}
    \caption{The prohibited pairs $(1, 3)$, $(2, 2)$ and $(3, 3)$ from the example in \Cref{fig:matchingexample1}.}
    \label{fig:matchingexample3}
\end{figure}

When the dents and bumps are so constructed, a tiling by modified sticks and staples should simulate a tiling with unmodified sticks that respects the matching rules. The proof of \Cref{lem:stapleverification} completes the proof of \Cref{thm:main}.

\begin{lemma}\label{lem:stapleverification}
The modified stick and staple tiles can tile the plane if and only if the unmodified stick tile can tile the plane while respecting the matching rules.
\end{lemma}

\begin{proof}
Given a tiling with unmodified sticks that respects the matching rules, we can replace each unmodified stick with the modified stick. This leaves some gaps that can be filled with the staple tile. Since the original tiling obeys the matching rules, the new tiling never has two bumps overlapping.

Suppose we have a tiling of the plane with modified sticks and staples. The staple tile has corners with inside angle $\frac{5\pi}{6}$. The corners of the unmodified stick tile have inside angles $\frac{2\pi}{3}$ and $\frac{4\pi}{3}$, but modifying the edges introduces some corners of angles $\frac{5\pi}{6}$, $\frac{7\pi}{6}$, $\frac{7\pi}{12}$, and $\frac{17\pi}{12}$. 

By \Cref{lem:nostapletiling}, a tiling of the plane with staples and modified sticks must have at least one stick. Align this stick with a hexagonal grid. Around each vertex, we must have vertices with angles that sum to $2\pi$ or vertices with angles summing to $\pi$ and a flat edge. Using angles $\pi$, $\frac{2\pi}{3}$, $\frac{4\pi}{3}$, $\frac{5\pi}{6}$, $\frac{7\pi}{6}$, $\frac{7\pi}{12}$ and $\frac{17\pi}{12}$, the only ways to reach exactly $2\pi$ are:

\[\pi+\pi\]
\[\frac{2\pi}{3}+\frac{2\pi}{3}+\frac{2\pi}{3}\]
\[\frac{2\pi}{3}+\frac{4\pi}{3}\]
\[\frac{5\pi}{6}+\frac{7\pi}{6}\]
\[\frac{7\pi}{12}+\frac{17\pi}{12}\]

So each of the unmodified corners of the stick must meet only unmodified corners of other sticks. If a given stick is aligned with the hexagonal grid, then the adjacent spaces in the grid are also contain sticks that are aligned with this grid. We see that all the sticks are aligned with the hexagonal grid, and every space in the grid is contains part of one stick. 

If we remove the staple tiles and replace each modified stick tile with an unmodified stick, then we obtain a tiling of the plane by unmodified sticks. This tiling must respect the matching rules, because otherwise the tiling with modified sticks would have had two bumps overlapping.
\end{proof}

So far, we have considered the case where tiles can be rotated but not reflected. We can modify the stick tile further so that a reflected tile and an unreflected tiling cannot be placed adjacent to each other in a tiling. An example is shown in \Cref{fig:prohibitingreflections}.

\begin{figure}
    \centering
    \includegraphics[page=6]{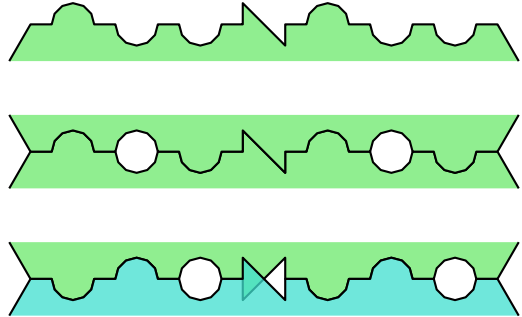}
    \caption{We can further modify each edge of the stick so that a reflected tile can never be placed next to a reflected tile. The edges can still be placed against each other if they have the same orientation, as in the middle figure. So even if reflections are allowed, all the stick tiles must occur with the same orientation.}
    \label{fig:prohibitingreflections}
\end{figure}

\subsection{Polyomino tiling}\label{sec:polyominos}

Ollinger (\cite{Ollinger5Tiling}) showed that tiling with $5$ polyomino prototiles is undecidable. Our construction relies on non-orthogonal rotations in order to form the weave pattern, and so it doesn't work for polyomino tiles. However, our techniques can be used to show that tiling with $3$ polyomino prototiles is undecidable. This is illustrated in \Cref{fig:polyominos}. We can construct an $n\times 1$ block tiling similar to the blocks in the schematics from \Cref{sec:schematics} and a $1\times 1$ gap tile that can be rotated $180$ degrees to represent part of either the left or right part of a gap. With appropriate edge-to-edge matching rules, these tiles can be made to create tilings according a set of gap conditions (see \Cref{sec:gapconditions}). Determining the existence of such a tiling is undecidable. A staple tile to simulate the edge-to-edge matching rules brings the total number of polyonimo tiles needed to $3$. 

\begin{figure}
    \centering
    \includegraphics[page=1]{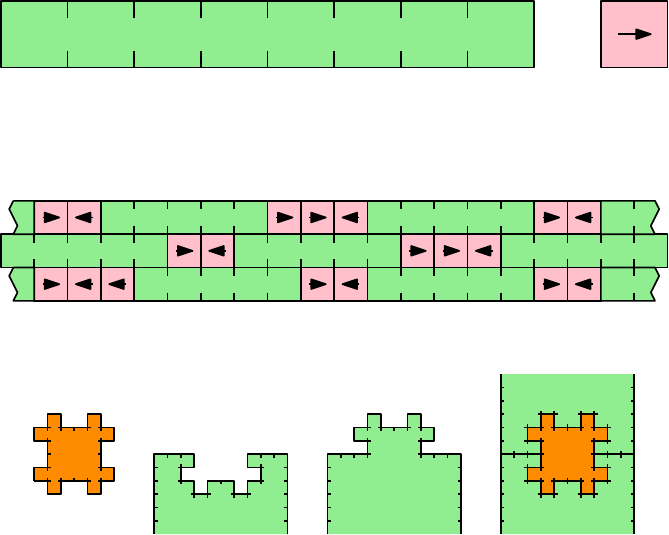}
    \caption{Top: a block tile and a gap tile. Middle: a section of a schematic made with block tiles and gap tiles. The gap tile occurs with $2$ different orientations. Bottom: a polyomino staple tile.}
    \label{fig:polyominos}
\end{figure}

While our tiles can't be polyominos, we could realize our construction with polyhexes, at the cost of making the staple tiles somewhat more complicated. This is illustrated in \Cref{fig:polyhex}. 

\begin{figure}
    \centering
    \includegraphics[page=2]{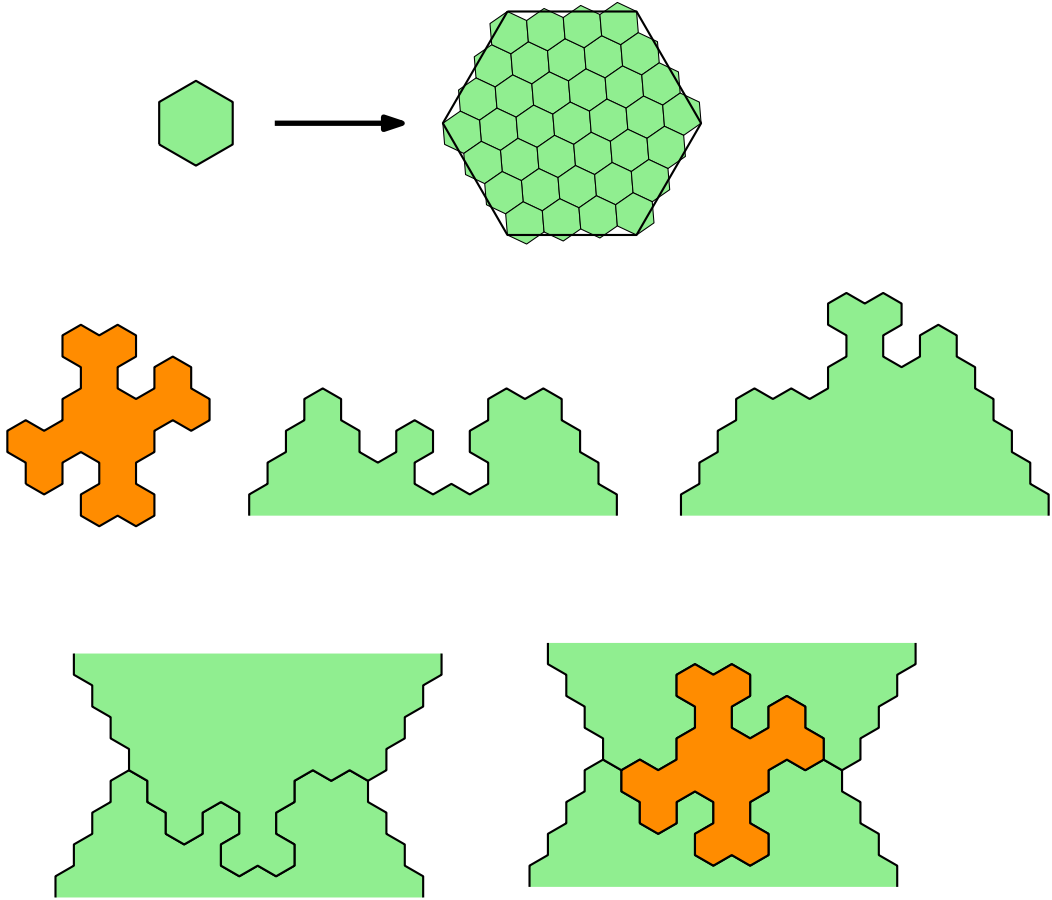}
    \caption{Realizing our construction with polyhexes. First, we need to increase the resolution of the stick tiles by replacing each hexagon by a large hexagon-shaped patch (top). Then we can add a staple tile made of polyhexes, and shown in the middle and bottom of the figure.}
    \label{fig:polyhex}
\end{figure}

\subsection{Translational tiling}

We considered tiling where rotations are allowed. It was previously known that translational $8$-tiling is undecidable (\cite{YangZhang8Tiling}). In our tilings, the stick tile occurs with 3 different orientations and the staple tile occurs with only one. So our results imply that translational 4-tiling is undecidable, and that translational 3-tiling is undecidable when there are edge-to-edge matching rules. 

The tiles from \Cref{sec:polyominos} also have the property that only 4 total orientations are used, so this shows that translation 4-tiling is undecidable even if the tiles are polyominos. A recent 

\section{Conclusion}\label{sec:conclusion}

An important open problem that remains is the decidability of monotiling. Aperiodic monotiles, that is prototiles that tile the plane but never with translational symmetry, have recently been found (\cite{SMKGSAperiodic1,SMKGSAperiodic2}). If monotiling is undecidable, then we would expect there to be infinitely many combinatorially different aperiodic monotiles, but so far at most a couple have been found.

\section{AI Statement}

Our Lean formalization was produced primarily by Harmonic's Aristotle tool, with smaller contributions by ChatGPT 5.6 Sol and ChatGPT 6 Astra. This version of the paper incorporates some minor bug fixes that were found during the formalization process. 

\section{Acknowledgments}

I am very grateful to Mikkel Abrahamsen and anonymous reviewers for feedback on earlier drafts of this article.

This work is supported by Independent Research Fund Denmark, grant 1054-00032B, and by the Carlsberg Foundation, grant CF24-1929.

\bibliographystyle{plainurl}
\bibliography{bib}

\end{document}